\documentclass[a4paper,11pt]{article}
\usepackage[left=.8in,right=.8in,top=1.1in,bottom=1.2in]{geometry}
\usepackage[T1]{fontenc}
\usepackage[utf8]{inputenc}
\usepackage{charter,mathpazo}
\usepackage[dvipsnames]{xcolor}
\usepackage{graphicx}
\usepackage{anyfontsize}
\usepackage[authoryear]{natbib}
\usepackage{caption}
\usepackage{subcaption}
\usepackage{amsmath,amssymb,bm}
\usepackage{amsthm}
		\newtheorem {proposition}{Proposition}	
\newtheorem {lemma}      {Lemma}	    
	\newtheorem {observation} {Observation}

\theoremstyle{definition}
\newtheorem*{definition*}{Definition}
\newtheorem {definition} {Definition}   

\newtheorem*{theorem*}{Theorem}
    	 	\newtheorem{remark}{Remark}      		
\theoremstyle{plain}

\usepackage{etoolbox}
\newcommand{\addQEDstyle}[2]{\AtBeginEnvironment{#1}{\pushQED{\qed}\renewcommand{\qedsymbol}{#2}}\AtEndEnvironment{#1}{\popQED}}
\addQEDstyle{example}{$\lozenge$}
\addQEDstyle{observation}{$\lozenge$}
\addQEDstyle{remark}{$\lozenge$}
\addQEDstyle{aobservation}{$\lozenge$}
\addQEDstyle{aremark}{$\lozenge$}
\addQEDstyle{fact}{$\lozenge$}
\usepackage{enumitem}
\setlist[enumerate]{label=(\alph*),itemsep=.2em,topsep=.3em}
\makeatletter

\newcommand\PDF[2]{\dfrac{\partial{#1}}{\partial{#2}}}
\newcommand{\relmiddle}[1]{\mathrel{}\middle#1\mathrel{}}
\newcommand\Mid{\relmiddle{|}}

\newcommand{\argmax}{\operatornamewithlimits{arg\,max}}

\makeatother
\allowdisplaybreaks \usepackage{mathtools}
\usepackage{mathrsfs} 
\usepackage{arydshln}
\usepackage{multirow}
\usepackage{setspace}
\usepackage[bookmarks=true,colorlinks=true,linkcolor=blue,urlcolor=blue,citecolor=blue,breaklinks=true,setpagesize=false,linktocpage]{hyperref}
\usepackage[noabbrev,capitalize]{cleveref}
\crefname{enumi}{}{}
\crefname{enumii}{}{}
\crefname{equation}{}{}

\usepackage[bottom]{footmisc}

\numberwithin{equation}{section}

\providecommand\FigureNote[1]{\vskip .5em \parbox[c]{\hsize}{\footnotesize \emph{Note:} #1}}
 \usepackage{amsmath,amssymb,bm}

\providecommand{\PDF}[2]{\frac{\partial{#1}}{\partial{#2}}}

\providecommand{\Vt}[1]{\bm{#1}}

\providecommand{\Vte}{\bm{e}}

\providecommand{\Vtm}{\bm{m}}

\providecommand{\Vtx}{\bm{x}}
\providecommand{\Vty}{\bm{y}}

\providecommand{\VtF}{\mathbf{F}}

\providecommand{\VtX}{\mathbf{X}}

\providecommand{\Vttheta}{\bm{\theta}}

\providecommand{\Vtpi}{\bm{\pi}}

\providecommand{\ClE}{\mathcal{E}}

\providecommand{\ClX}{\mathcal{X}}
\providecommand{\ClY}{\mathcal{Y}}

\providecommand{\BbR}{\mathbb{R}}

\providecommand{\BbZ}{\mathbb{Z}}

\providecommand{\Brx}{\bar{x}}

\providecommand{\Bralpha}{\bar{\alpha}}

\providecommand{\Tlpi}{\tilde{\pi}}

\providecommand{\Dtx}{\dot{x}}

\providecommand{\Is}{\coloneq}

\providecommand{\BrVtx}{\bar{\Vtx}}

\providecommand{\TlVtpi}{\tilde{\Vtpi}}

\providecommand{\inX}{\in\ClX}
\providecommand{\inY}{\in\ClY}

\begin{document}

\onehalfspacing 
\title{Most Likely Retail Agglomeration Patterns: Potential Maximization and Stochastic Stability of Spatial Equilibria\thanks{We thank Masashi Suzuki and Shuhei Yamaguchi for their generous research assistance. Minoru Osawa thanks the grant support from JSPS Kakenhi 17H00987, 19K15108, and 22H01610.}} 
\author{Minoru Osawa,\thanks{Corresponding Author. \href{mailto:osawa.minoru.4z@kyoto-u.ac.jp}{osawa.minoru.4z@kyoto-u.ac.jp}; Institute of Economic Research, Kyoto University, Yoshida-honmachi, Sakyo-ku, Kyoto, Kyoto 606-8501}\ \ \ Takashi Akamatsu,\thanks{Corresponding Author. \href{mailto:akamatsu@plan.civil.tohoku.ac.jp}{akamatsu@plan.civil.tohoku.ac.jp}; Graduate School of Information Sciences, Tohoku University, 06 Aramaki-Aoba, Aoba-ku, Sendai, Miyagi 980-8579, Japan.}\ \ \ and Yosuke Kogure\thanks{\href{mailto:kogure@gipc.akita-u.ac.jp}{kogure@gipc.akita-u.ac.jp}; Graduate School of Engineering Science, Akita University, 1-1 Tegatagakuen-machi, Akita City 010-8502, Japan.}}
\date{\today}
\maketitle
\begin{abstract}
\noindent 
We study a model of retail agglomeration where consumers are more likely to visit zones with a higher concentration of shops. This agglomerative effect makes zones with many retailers more attractive. The spatial distribution of retailers in equilibrium is endogenously determined in response to the spatial pattern of shopping demand. In such a setting, multiple locally stable equilibria may arise, and the outcome can depend on the initial distribution of shops. To address this issue, we apply an approach from evolutionary game theory, selecting the equilibrium that maximizes a potential function representing the incentives of retailers. We demonstrate the method in a two-dimensional spatial setting. Compared to local stability based on gradual, myopic adjustments, this global maximization leads to a unique and more robust prediction. As expected, the number of retail clusters decreases either when shopping costs for immobile consumers fall or when the attractiveness of larger retail concentrations increases. 
\end{abstract}

\medskip
\noindent\textbf{Keywords:} 
retail agglomeration; 
spatial interaction;  
multiple equilibria;
local stability; 
stochastic stability. 

\medskip
\noindent\textbf{JEL Classification:} C62, R12, R13, R14

\clearpage

\section{Introduction}
\label{sec:introduction}

The key features of cities include activities at locations, spatial interaction patterns between these locations, and the spatial distribution of various factors that facilitate these activities and interactions. 
\cite{Wilson-RSI2007} describes a general framework called the Boltzmann--Lotka--Volterra (BLV) method to model such systems. The BLV method is a synthesis of \emph{fast dynamics} (the ``Boltzmann'' component) and \emph{slow dynamics} (the ``Lotka--Volterra'' component). 
Fast dynamics determine the spatial interaction flows between locations in the short run (e.g., trade flows, commuting flows) based on the entropy-maximizing framework \citep{wilson1967statistical}. 
The spatial distribution of mobile actors (e.g., firms, households) is treated as fixed in the short run, and governs the generation and attraction of such flows. 
The spatial interaction flows determine the short-run payoff landscape for mobile actors (e.g., profits for firms, utilities for households). 
Slow dynamics then describe the gradual adjustments of the spatial distribution of mobile actors by considering their relocation. 
This approach to combine short- and long-run dynamics was followed in the so-called ``new economic geography'' \citep{Fujita-Krugman-Venables-Book1999,Baldwin-Book2016} and paved the way for quantitative spatial models in economics \citep[see][for a survey]{Redding-Rossi-Hansberg-ARE2017}. 

The first instance of the BLV method is the \cite{Harris-Wilson-EPA1978} (HW) model, a pioneering work in modeling the spontaneous formation of retail agglomerations in an urban area. Based on the static shopping models of \cite{huff1963probabilistic} and \cite{lakshmanan1965retail}, HW formulated a spatial model with agglomeration and dispersion forces. 
Retail firms tend to agglomerate in fewer locations because consumers are attracted to larger concentrations of retailers. 
They may also prefer to disperse spatially to be closer to consumers and to avoid competition from other retail agglomerations. 
HW demonstrated that such a model can exhibit multiple equilibria, path dependence, and catastrophic phase transitions. 
Although their study focused on retail agglomerations in urban areas as an application, the HW model (and more broadly, the BLV methodology) has since been applied to a wider range of fields, including logistics \citep{leonardi1981unifying1,leonardi1981unifying2}, archaeology \citep{bevan2013models,paliou2016evolving}, healthcare \citep{tang2017flow}, and crime \citep{davies2013mathematical,baudains2016dynamic} to name a few \citep[see][for a survey]{dearden2015explorations,Wilson-NETS2024}. 

In light of the diverse applications, it is essential to have a solid understanding of the analytical aspects of the HW model.
One key issue with the HW model is the multiplicity of equilibria.
The early explorations of the analytical properties of the model mainly focused on two-location settings for tractability \citep{clarke1981note,rijk1983equilibrium,rijk1983uniqueness}. However, these studies had already suggested that multiple locally stable equilibria could arise.
Thereafter, extensive numerical simulations of the HW model in multi-location settings have demonstrated that numerous locally stable equilibria can arise for each set of parameter values \citep{clarke1983dynamics,clarke1985dynamics,wilson2011phase,dearden2015explorations}.
Furthermore, using an analytical method developed by \cite{Akamatsu-Takayama-Ikeda-JEDC2012}, \cite{Osawa-et-al-JRS2017} formally demonstrated that the model allows multiple locally stable states in multi-location settings.
This implies that the model's predictions can be elusive and raises questions about the robustness of numerical findings in the literature. 

To address this issue, we introduce a new approach that enables the unambiguous prediction of the most likely spatial configuration in the HW model. 
We employ the results from the theory of \emph{potential games} \citep{monderer1996potential,Sandholm-JET2001,Sandholm-JET2009}. 
We first establish that the HW model is a potential game, that is, its structure can be represented by a single function called the \emph{potential function} that assigns a scalar for each possible state (the spatial distribution of retailers). 

In potential games, various analytical tools are available for the characterization of the properties of equilibria. 
First, the set of Nash equilibria of a potential game coincides with that of Karush--Kuhn--Tucker (KKT) points of the maximization of the potential function. 
Second, the set of local maximizers of the potential function is \emph{locally stable} under various standard dynamics.\footnote{\label{foot:dynamics}For example, local potential maximizers are known to be locally stable under 
the \emph{best response dynamic} \citep{Gilboa-Matsui-ECTA1991}, 
the \emph{Brown--von Neumann--Nash dynamic} 
\citep{Brown-vonNeumann-Rand1950,Nash-AM1951}, 
the \emph{Smith dynamic} \citep{smith1984stability}, 
and \emph{Riemannian game dynamics} \citep{Mertikopoulos-Sandholm-JET2018} such as the replicator dynamics \citep{Taylor-Jonker-MB1978} often employed in new economic geography.  
See \cite{Sandholm-JET2001} and \cite{Sandholm-Book2010}, Section 8.2.}  
Third, global maximizers of the potential function are \emph{globally stable} in multiple senses. 
The set of global potential maximizers in a potential game is selected as ``stochastically stable'' states \citep[][Section 12.2]{Sandholm-Book2010}, roughly meaning that such states are most likely to persist in the long run when small random perturbations can occur. 
Another approach is dynamic optimization, where perfect foresight dynamics select the global maximizer(s) of the potential function \citep{Oyama-RSUE2009,oyama2009history}. 

In particular, by global maximization of the potential function, we can determine the most likely spatial agglomeration patterns at each given value of the structural parameters of the model. This contrasts with the local stability approach in the literature, which often leaves multiple stable equilibria. 
To demonstrate the effectiveness of this approach, we analyze several stylized examples (two zone city and a two-dimensional city). 
We show that, in the most likely spatial configuration, the number of retail agglomerations decreases either when shopping costs for consumers decrease or when the strength of agglomerative effects increases. 
These results corroborate numerical findings in the literature. 

\section{Related literature}
\label{sec:literature}

\cite{Harris-Wilson-EPA1978} showed that their model reduces to a maximization problem of a scalar-valued function with respect to consumers' shopping patterns and spatial distribution of retail firms, and interpreted it as a welfare maximization problem of a central planner. 
We instead interpret the HW model as a large-population potential game \citep{Sandholm-JET2001,Sandholm-JET2009}, 
which, in turn, allows us to employ the theory of potential games to analyze the model.  
As we have discussed above and will discuss in \cref{sec:potential_game}, we employ (global) potential maximization as an equilibrium refinement criterion.  
The approach is motivated by \cite{Sandholm-Book2010}'s \emph{stochastic stability} method (Sections 11 and 12) that considers the limiting behaviors of probability distributions over the set of possible spatial patterns of retailers. 
Specifically, Sandholm considered the ``stationary distribution'' of stochastic evolutionary process, which describes the probability of a spatial pattern to occur in the long run. 
In a potential game, the stationary distribution assigns a higher probability at a spatial pattern that achieve a higher potential function value. 
When stochasticity diminishes, the distribution concentrates on the set of global maximizers of the potential function \citep{Blume-GEB1993,blume1997population}, and such states are considered to be \emph{stochastically stable}. 

One successful application of stochastic stability and potential games in the urban spatial context is \cite{schelling1971dynamic}'s model of segregation, in which a finite number of agents choose their locations on a discrete grid taking neighbors' characteristics into account. 
Schelling's studies demonstrate that a small microscopic homophily can lead to macroscopic separation of two groups of people. 
By considering several specific functional forms for an individual's utility function, \cite{zhang2004dynamic,zhang2004residential,zhang2011tipping} showed that potential functions can be used to characterize the equilibria of the model. 
Building on Zhang's work, \cite{Grauwin-etal-JPE2012} formulated Schelling's model as a spatial evolutionary game and provided a general analysis using a potential game method. 
\cite{zhang2004dynamic,zhang2004residential,zhang2011tipping} and \cite{Grauwin-etal-JPE2012} considered games with a finite number of agents and studied the limiting behavior of the stationary distribution when stochasticity diminishes. 

Some studies have focused on the behavior of the stochastic differential equation (SDE) based on the HW model. 
\cite{vorst1985stochastic} considered an SDE version of the HW model and defined a different type of stochastic stability concept, showing that equilibrium in the model is globally absorbing if it is unique. 
Our study instead focuses on the cases in which the model features multiple equilibria. 
More recently, \cite{Ellam-etal-PRSA2018} proposed an approach based on a SDE formulation for the HW model, and provided a Bayesian method for parameter estimation. 
They also exploit a potential function associated with their version of the HW model. The stationary distribution associated with their SDE is then represented by the potential function, and forms an integral part of their parameter estimation procedure. 
Our study is different from theirs in that we focus on the game-theoretic interpretation of the HW model, and we aim to demonstrate the effectiveness of potential maximization as equilibrium refinement for deterministic spatial models using the HW model as a concrete example. 
Also, it is noted that the theoretical foundation behind our potential maximization approach (i.e., stochastic stability) lies in considering individual noises on the side of retailers choices. 
However, stochasticity in \cite{Ellam-etal-PRSA2018} arises from aggregate fluctuations. 
The synthesis of the two approaches is an interesting avenue for future research. 

We contribute to the literature on potential game methods in spatial economic models.
One strand of this literature analyzes the formation of central business districts as a result of agents' social preferences for proximity to others, as originally proposed by \citet{beckmann1976spatial} and later revisited by \citet{Mossay-Picard-JET2011}.
By generalizing the framework of \citet{Mossay-Picard-JET2011} using Beckmann-type social externalities, \citet{Blanchet-et-al-IER2016} developed a variational (potential maximization) formulation for a broad class of urban spatial models with a continuum of agents in continuous space. 
Characterizing the stability of equilibria in such continuous-space models is challenging, although some attempts have been made in this direction \citep[e.g.,][]{Bragard-Mossay-CSF2016}. 
An alternative approach is to consider discrete-space versions of Beckmann-type models, as in \citet{Akamatsu-JME2017}, who employed tools from finite-strategy potential games with continuum players \citep{Sandholm-Book2010}. 
Along similar lines, \citet{Osawa-Akamatsu-2020} showed that the seminal model of multiple business district formation by \citet{Fujita-Ogawa-RSUE1982} can be interpreted as a potential game when reformulated in discrete space, where global maximization of the potential function serves as a powerful analytical tool.
The discrete-space approach also offers a tractable strategy for analyzing continuous-space models, including extensions of Schelling-type frameworks with a continuum of agents and continuous space, such as the one proposed by \citet{Mossay-Picard-JRS2019}.
A complementary direction is to develop a general theory of large-population potential games with continuous strategy sets, as pursued by \citet{Cheung-Lahkar-GEB2018, Lahkar-Riedel-GEB2015, Cheung-JET2014, Cheung-GEB2016}.
As an application of such theories in urban economics, \citet{Takayama-ET2020} examined a monocentric city model with bottleneck congestion and characterized its equilibria using a potential maximization approach.
We expect that potential game methods will continue to provide fruitful insights into urban and regional economic problems.

\section{The model}
\label{sec:model}

In the following, $\BbZ$ denotes the set of integers. 
$\BbR$ and $\BbR_+$ denote the set of reals and nonnegative reals, respectively. 
For $n\in\BbZ$, we define $[n] \Is \{1, 2, \hdots, n\}$. 

Consider a city comprising $K\in\BbZ$ discrete zones. 
Let $[K]$ denote the set of zones. 
There is a large continuum of retailers that can enter or exit any zone in the city. 
The mass of retailers in the city is endogenously determined in equilibrium.  
The spatial distribution of retailers is denoted by $\Vtx = (x_i)_{i\in[K]}$, where $x_i \ge 0$ is the mass of retailers in zone $i\in[K]$. 
We call $i\in[K]$ with $x_i > 0$ a \textit{retail agglomeration}. 
The spatial distribution of the retailers, $\Vtx$, is the endogenous variable of the model. 

There is a continuum of consumers whose spatial distribution is exogenously given. 
Each infinitesimal consumer purchases a single unit of goods sold by retailers. 
In aggregate, consumers' shopping behavior is modeled by a set of origin-constrained gravity equations, which is originally derived from the entropy-maximization principle \citep{wilson1967statistical}. 
The value spent in zone $i$ by consumers in zone $j$ is given as 
\begin{align}
    V_{ji}(\Vtx) = 
        \frac
            {x_i^\alpha \exp\left( - \beta t_{ji}\right)}
            {\sum_{k\in[K]} x_k^\alpha \exp\left( - \beta t_{jk}\right)}
        Q_j, 
\end{align}
where the fixed constant $Q_j > 0$ denotes the total demand (or the mass of consumers) in zone $j$,   
$t_{ji}$ is the generalized travel cost from zone $j$ to $i$, and 
$\beta > 0$ represents the rate at which demand decreases in distance. 
Also, in HW's terminology, $x_i^\alpha$ is the \textit{attractiveness} of zone $i$ for consumers, and $\alpha > 1$ is the elasticity of attractiveness. Greater $\alpha$ implies stronger agglomeration effects. 

The total revenue of zone $i$ is equally distributed among the active retailers therein. 
Each active retailer in zone $i$ incurs a constant cost $\kappa_i > 0$ to operate. 
The profit of a retailer in zone $i$, which is a function of $\Vtx$, is then given by 
\begin{align}
    \pi_i(\Vtx) 
        & 
            \Is 
            \frac{1}{x_i} \sum_{j \in[K]} V_{ji} (\Vtx) - \kappa_i
            = 
            \sum_{j \in[K]}
            \frac  
                { x_i^{ \alpha - 1 } \exp\left( - \beta t_{ji}\right)}
                {\sum_{k\in[K]} x_k^\alpha \exp\left( - \beta t_{jk}\right)}
            Q_j
            -
            \kappa_i. 
\end{align}
It is noted that $\alpha > 1$ ensures that $\pi_i$ is well-defined for all $\Vtx \ge \Vt0$. 

Equilibrium spatial distribution of retailers is determined by their entry--exit behavior. 
Retailers may not enter the city if they do not ern a nonnegative profit. 
Let $\pi_0(\Vtx)$ be the payoff of the outside option for retailers (i.e., not entering the city) and let $\pi_0(\Vtx) = 0$ for any $\Vtx$. 
Retailers enter zone $i$ if $\pi_i(\Vtx) > \pi_0 = 0$, exit if $\pi_i(\Vtx) < \pi_0 = 0$. 
In equilibrium, retailers in all zones achieve zero profit: $x_i > 0$ implies $\pi_i(\Vtx) = 0$, and $\pi_i(\Vtx) = 0$ implies $x_i \ge 0$. 
Also, there is no incentive for retailers to enter zones without retailers because $x_i = 0$ implies $\pi_i(\Vtx) = - \kappa_i <\pi_0 =  0$. 
To sum up, spatial equilibrium can be defined as follows. 
\begin{definition}
    \label{def:equilibrium}
    A spatial distribution of retailers $\Vtx\ge\Vt0$ is a \emph{spatial equilibrium} if it satisfies the following condition: 
    \begin{align}
        & x_i \pi_i(\Vtx) = 0,\ x_i \ge 0,\ \pi_i(\Vtx) \le 0 
            & \forall i\in[K]. 
        \label{eq:equilibrium}
    \end{align}
\end{definition}

\begin{lemma}
    \label{lem:X}
Any spatial equilibrium must lie in the following closed and convex set:
\begin{align}
    \ClX 
        \Is 
        \left\{
            \Vtx \in \BbR_+^K 
        \Mid
            \sum_{i\in[K]} \kappa_i x_i = Q
        \right\},
    \label{eq:X}
\end{align}
where $Q \Is \sum_{i\in[K]} Q_i$ is the total mass of consumers or the total retail demand in the city. 
\end{lemma}
\begin{proof}
We have $\sum_{i\in[K]} x_i \pi_i(\Vtx) = 0$ at any spatial equilibrium, implying $\sum_{i\in[K]} Q_i = \sum_{i\in[K]} \kappa_i x_i$. 
\end{proof}

The left-hand side of \cref{eq:X} is the retailers' total revenue, whereas the right-hand side is their total cost in the city. 
We assume that there is a sufficiently large pool of possible entrants, so that the total mass of retailers equals the total demand at any spatial equilibrium. 

\section{Multiplicity of locally stable equilibria}
\label{sec:multiplicity}

Following the literature, we study the comparative statics of spatial equilibrium with respect to the structural parameters of the model. 
In particular, we focus on the roles of $\alpha$ and $\beta$. 

It is known that the HW model can have numerous spatial equilibria. 
For example, \cite{rijk1983uniqueness} showed that there are \emph{at least} $\binom{K}{K/2} + 1$ spatial equilibria when $K$ is even. 
To obtain relevant outcomes among multiple equilibria, the literature focuses on locally stable equilibria under the following natural deterministic dynamics: 
\begin{align}
    & \Dtx_i 
            = x_i \pi_i(\Vtx) 
            = 
            \sum_{j \in[K]}
            \frac  
                { x_i^{ \alpha } \exp\left( - \beta \ell_{ji}\right)}
                {\sum_{k\in[K]} x_k^\alpha \exp\left( - \beta \ell_{jk}\right)}
            Q_j
            -
            \kappa x_i
    & \forall i\in[K]. 
    \label{eq:D}
    \tag{D}
\end{align}
The dynamic is consistent with the equilibrium condition \cref{eq:equilibrium} in that any spatial equilibrium is a stationary point of \cref{eq:D}. 
In the terminology of evolutionary game theory \citep[as surveyed by][]{Sandholm-Book2010}, the dynamic \cref{eq:D} is a special case of the replicator dynamic \citep{Taylor-Jonker-MB1978} where the average payoff is always zero because $\pi_0 = 0$.

Assuming \cref{eq:D}, we can focus on the states in $\ClX$ since $\ClX$ defined by \cref{eq:X} is globally attracting on $\BbR_+^K$. 
In fact, for any $\Vtx \ge \Vt0$, 
\begin{align}
    \sum_{i\in[K]} \Dtx_i = \sum_{i\in[K]} x_i \pi_i(\Vtx) = Q - \sum_{i\in[K]} \kappa_i x_i, 
\end{align} 
so that the total mass of retailers strictly increases if $Q > \sum_{i\in[K]} \kappa_i x_i$ and strictly decreases if $Q < \sum_{i\in[K]} \kappa_i x_i$. 
Thus, $\sum_{i\in[K]} \kappa_i x_i = Q$ at any stationary point of \cref{eq:D}. 

\begin{figure}[tb]
    \centering
    \begin{subfigure}[b]{.16\hsize}
        \centering
        \includegraphics[width=2.1cm]{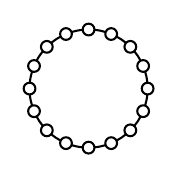}
        \caption{$K = 2^4$ circle} 
    \end{subfigure}
    \begin{subfigure}[b]{.16\hsize}
        \centering
        \includegraphics[width=2.1cm]{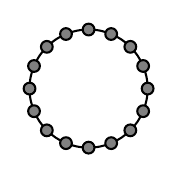}
        \caption{$\Vtx^{(0)}$} 
    \end{subfigure}
    \begin{subfigure}[b]{.16\hsize}
        \centering
        \includegraphics[width=2.1cm]{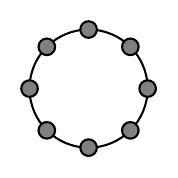}
        \caption{$\Vtx^{(1)}$} 
    \end{subfigure}
    \begin{subfigure}[b]{.16\hsize}
        \centering
        \includegraphics[width=2.1cm]{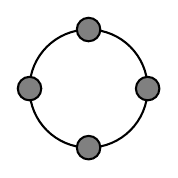}
        \caption{$\Vtx^{(2)}$} 
    \end{subfigure}
    \begin{subfigure}[b]{.16\hsize}
        \centering
        \includegraphics[width=2.1cm]{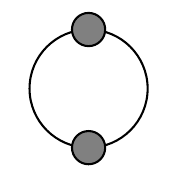}
        \caption{$\Vtx^{(3)}$} 
    \end{subfigure}
    \begin{subfigure}[b]{.16\hsize}
        \centering
        \includegraphics[width=2.1cm]{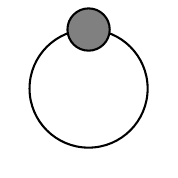}
        \caption{$\Vtx^{(4)}$} 
    \end{subfigure}
    \caption{Circular economy and symmetric spatial configurations.}
    \FigureNote{Gray disks indicate the masses of retailers in each zone.} 
    \label{fig:REs}
\end{figure} 
Equilibrium refinement based on local stability under \cref{eq:D} can leave numerous equilibria as locally stable states. 
In particular, the mono-centric concentration of retailers in any single zone is always locally stable under \cref{eq:D} if $\alpha > 1$ (see \cref{app:proofs} for the omitted proofs). 
\begin{proposition} 
\label{prop:monnopolar}
If $\alpha > 1$, a full concentration of retailers in any single zone, namely $x_i = Q/\kappa_i$ and $x_j = 0$ ($j\ne i$) for some $i\in[K]$, is a locally stable spatial equilibrium under \cref{eq:D}. 
\end{proposition}

That is, there is at least as many locally stable equilibria as the number of zones at any level of transport costs between locations. 
Further, the following result demonstrates that spatial patterns with more than two retail agglomerations can become locally stable simultaneously. 

\begin{proposition}
Let $\kappa_i = \kappa$ for all $i\in[K]$ and also assume that $\frac{Q}{\kappa} = 1$, so that $\ClX$ is the $(K-1)$-simplex. 
Suppose $\alpha > 1$. 
Consider a one-dimensional symmetric circular economy, where $\ell_{ij} = \min\{|i - j|, K - |i - j|\}$ and $\kappa_i = 1$ for all $i\in[K]$. 
Assume $K = 2^J$ with $J \ge 3$. 
If $\alpha$ is sufficiently small and $\beta$ is sufficiently large, all spatial patterns of the form 
\begin{align}
    \Vtx^{(k)} \Is (\underbrace{\underbrace{2^k \Brx, 0, 0,\hdots,0}_{\text{$2^k$ elements}}, 
        \underbrace{2^k \Brx, 0, 0,\hdots,0}_{\text{$2^k$ elements}}, 
        \hdots,
        \underbrace{2^k \Brx, 0, 0,\hdots,0}_{\text{$2^k$ elements}}}_{\text{repeated $K/2^k = 2^{J - k}$ times}})
\end{align}
with $0 \le k \le J$ and $\Brx = \frac{1}{K}$, up to symmetry, are locally stable spatial equilibrium under \cref{eq:D}.
\end{proposition}
\begin{proof}
See Proposition 4 as well as Figure 8 of \cite{Osawa-et-al-JRS2017}. 
\end{proof}
\noindent If $\Vtx^{(k)}$ ($k \ge 2$) is locally stable, then \emph{all} $\Vtx^{(l)}$ ($l\ge k$) are locally stable. 
For the case of $K = 2^4 = 16$, \cref{fig:REs} shows the circular economy and spatial patterns $\{\Vtx^{(k)}\}_{1 \le k \le J}$. 
Furthermore, there can be other spatial patterns that are locally stable under \cref{eq:D}.

\section{Potential and stability} 
\label{sec:potential_game}

The equilibrium refinement based on local stability can thus leave multiple equilibria. 
To overcome this issue, this section introduces a different approach based on potential game theory. 

\subsection{The HW model as a large-population potential game}

We observe that the HW model is a \textit{large-population potential game}. 
First, large-population games are defined as follows \citep[see][for a survey]{Sandholm-Book2010}.  
\begin{definition}[Large-population game]
    \label{def:large_population_game}
Consider a game played by a continuum of homogeneous agents. 
Let $[S]$ be the set of available discrete actions, where $S\in\BbZ$ is the number of actions. 
Let $y_i \in [0,1]$ be the share of agents that play action $i\in[S]$ and $\ClY \Is \{\Vty \in\BbR_+^S \mid \sum_{i\in[S]} y_i = 1\}$ be the set of all possible action distributions over the continuum population. 
Let $\VtF$ be the $\BbR^S$-valued Lipschitz continuous payoff function, whose $i$th component $F_i$ maps a state $\Vty\inY$ to payoff $F_i(\Vty)$ for agents choosing action $i\in[S]$ at the state. 
The tuple $([S],\VtF)$ is called a \textit{large-population game}. 
\end{definition}

\begin{definition}
A \emph{Nash equilibrium} of a large-population game $([S],\VtF)$ is a state $\Vty^*\inY$ that satisfies the following condition: 
$y_i^* > 0 \Rightarrow i \in \argmax_{k\in [S]} F_k(\Vty^*)$,
or equivalently, 
\begin{align}
    & y_i^* \left(v^* - F_i(\Vty^*)\right) = 0, y_i^* \ge 0, v^* \ge F_i(\Vty^*) 
    \quad \text{where} \quad v^* \Is \max_{k\in[K]} F_k(\Vty^*).
    \label{eq:NE}
\end{align}
\end{definition}

A large-population game whose payoff function can be characterized by a single scalar-valued function is called a large-population potential game \citep{Sandholm-JET2001,Sandholm-JET2009}. 
\begin{definition}[Large-population potential game]
\label{def:potential_game}
A large population game $([S], \VtF)$ is a \textit{potential game} if there is a scalar-valued function $f$ defined in the neighborhood of $\ClY$ that satisfies $\frac{\partial f(\Vty)}{\partial y_i} = F_i(\Vty)$ for all $i\in[S]$ and $\Vty\inY$. 
\end{definition} 

The next observation follows. 
\begin{observation}
The HW model is a large-population potential game. 
\end{observation}

The HW model is a large-population game $(\{0\}\cup [K],\Vtpi)$, where $\{0\}\cup [K]$ is the set of retailers' possible actions including the outside option $0$ (not entering any zone), and the payoff function is $\Vtpi(\Vtx) \Is (\pi_0(\Vtx), \pi_1(\Vtx), \hdots, \pi_K (\Vtx))$ with $\pi_0(\Vtx) = 0$ for all $\Vtx$. 
Also, the equilibrium condition \cref{eq:equilibrium} is equivalent to the Nash equilibrium condition \cref{eq:NE} with $v^* = 0$ ($ = \pi_0(\Vtx) $), and $x_0 = X - \sum_{i\in[K]} \kappa_i x_i > 0$ with $X$ being a sufficiently large total mass of possible entrants. 
That is, $x_0$ is the mass of inactive retailers who are choosing not entering the city.\footnote{The state variable may be redefined to $\Vty = (x_0/X, x_1/X, \hdots, x_K/X)$ to be fully consistent with \cref{def:large_population_game}. Since we will not use such $\Vty$ afterward, with a slight abuse of notation we keep using $\Vtx = (x_i)_{i\in[K]}$ as the state variable.}  
Furthermore, the following function is the potential function for $\Vtpi(\Vtx)$:\footnote{In a different context, \cite{Ellam-etal-PRSA2018} introduces a scalar-valued ``potential function'' associated with a stochastic version of the HW model. \cref{app:ellam} discusses how their potential function is related to ours.} 
\begin{align}
\label{eq:HW-P}
    & 
    f(\Vtx) 
        \Is 
        A(\Vtx)
        -
        \sum_{i\in[K]} \kappa_i x_i. 
        \tag{P},  
\end{align}
where 
\begin{align}
    A(\Vtx) \Is \frac{1}{\alpha} \sum_{j\in[K]} Q_j \log \left(\sum_{k\in[K]} x_k^\alpha \exp\left( - \beta \ell_{jk}\right) \right). 
\end{align}
We can confirm $\frac{\partial f(\Vtx)}{\partial x_i} = \pi_i(\Vtx)$ for all $i\in [K]$ and $\frac{\partial f(\Vtx)}{\partial x_0} = 0 = \pi_0(\Vtx)$, satisfying the definition of a potential function (\cref{def:potential_game}). 

In the potential function $f$, the second term is simply the total cost for retailers. 
The first term, $A$, can be interpreted as a welfare measure for immobile consumers because it is the aggregate accessibility to retail agglomerations \citep{Harris-Wilson-EPA1978,leonardi1978optimum}.  
Concretely, suppose that each consumer in $i\in[K]$ chooses their shopping destination $j\in[K]$ by maximizing the random utility of the form $u_{ij} = \alpha\log x_j - \beta t_{ij} + \epsilon$ with $\epsilon$ being i\@.d\@.d\@. Gumbel random variable.  
Then, $A$ is the city-wide aggregate of a log-sum function commonly used in transport research \citep[see][for a survey]{deJong-etal-TRA2007}. 
For any two equilibria $\Vtx^*$ and $\Vtx^{**}$ in $\ClX$, we have $\sum_{i\in[K]} \kappa_i x_i^* = \sum_{i\in[K]} \kappa_i x_i^{**} = Q$. 
Thus, we have 
$f(\Vtx^*) - f(\Vtx^{**})
    = 
    A(\Vtx^*)
    - 
    A(\Vtx^{**})$, implying the following:  
\begin{observation}
When multiple equilibria exist, the one with the higher potential function value offers greater aggregate accessibility from the perspective of immobile consumers. 
\end{observation}

\begin{remark}
We confirm that $\Vtpi$ satisfies \emph{externality symmetry}, a necessary and sufficient condition for the existence of a potential function as described in \cref{def:potential_game} \citep{Sandholm-JET2001}.\footnote{A less demanding definition of potential games than \cref{def:potential_game} and the associated symmetry requirement for $\Vtpi$ can be found in \cite{Sandholm-JET2009}.} 
Specifically, $\Vtpi$ satisfies 
\begin{align}
    & \PDF{\pi_i(\Vtx)}{x_j} = \PDF{\pi_j(\Vtx)}{x_i} 
    & \forall i,j, \Vtx. 
    \label{eq:symmetry}
\end{align} 
That is, the marginal increase of the payoff of retailer $i$ when the mass of retailers in zone $j$ increases is the same as the marginal increase of the payoff of retailer $j$ when the mass of retailers in zone $i$ increases. 
In fact, if either $i = 0$ or $j = 0$, we confirm that both sides of \cref{eq:symmetry} are zero so the equality holds true. 
For $i,j\in[K]$, we confirm that
\begin{align}
    & \PDF{\pi_i(\Vtx)}{x_j} 
    = 
    \alpha
    \sum_{k\in[K]} 
    S_{ki} 
    S_{kj}
    Q_k   
    = 
    \alpha
    \sum_{k\in[K]} 
    S_{kj}
    S_{ki} 
    Q_k 
    = 
    \PDF{\pi_j(\Vtx)}{x_i}
    & \forall i,j \in[K], \forall \Vtx \inX, 
    \label{eq:payoff-symmetry}
\end{align}
where $S_{ik}$ is the share of a retailer in zone $i$ over the demand from consumers in zone $k$: 
\begin{align}
    S_{ki} \Is \frac{x_i^{\alpha - 1} \exp\left(-\beta \ell_{ki}\right)}{\sum_{l\in[K]} x_l^\alpha \exp\left(-\beta \ell_{kl}\right)}.
\end{align} 
The symmetry condition \cref{eq:payoff-symmetry} indicates that the competition over the shopping demand from each zone $k$ is symmetric between marginal entrant in zones $i$ and $j$, in the sense that the marginal increase in the level of competition is the product of their shares $S_{ik}$ and $S_{jk}$. 
\end{remark}

\subsection{Potential maximization and stochastic stability of equilibria}

\label{sec:how-to-apply-ss}

Knowing that the HW model is a potential game, we can apply the potential maximization method to analyze its equilibria. 
Since all spatial equilibria of the HW model are contained in $\ClX$, the following potential maximization problem characterizes the equilibria of the HW model: 
\begin{align}
    & \max_{\Vtx\inX}.\ f(\Vtx). 
    & \label{eq:potential_maximization}
    \tag{PM}
\end{align} 
That is, the first-order necessary condition for the extrema (the KKT condition) for \cref{eq:potential_maximization} is equivalent to the equilibrium condition \cref{eq:equilibrium} \citep{Sandholm-JET2001,Sandholm-JET2009}. 

For example, if $\alpha \in (0,1)$, $f$ is strictly concave, so that \cref{eq:potential_maximization} has a \emph{unique} global maximizer and equilibrium is unique \citep[cf.][Theorem 1]{vorst1985stochastic}. 
This study instead focuses on the case $\alpha > 1$ where multiple equilibria can exist. 
We can associate the local and global maximization of $f$ with two types of equilibrium refinements.  

Local maximization of the potential function is closely linked with local stability under deterministic dynamics.  
\cite{Sandholm-JET2001} showed that the set of local potential maximizers coincides with that of locally stable states under various myopic evolutionary dynamics including \cref{eq:D} (see \cref{foot:dynamics} on \cpageref{foot:dynamics}). \cref{prop:monnopolar} corresponds to the fact that each corner of $\ClX$, which is the full concentration of retailers in a single zone, is a local maximizer for the problem \cref{eq:potential_maximization} if $\alpha > 1$, and hence locally stable. 
As we have discussed in \cref{sec:multiplicity}, however, the refinement based on local stability can leave multiple equilibria and it is agnostic about which equilibria are ``more relevant.'' 

Focusing instead on the global potential maximizers can obviously provide a stronger means of equilibrium refinement because the set of global maximizers are often a proper subset of that of local maximizers. 
In fact, there are rigorous economic foundations behind global maximization of the potential function. 
One of the most important approaches is the \emph{stochastic stability} method based on stochastic evolutionary dynamics \citep[][Section 12.2]{Sandholm-Book2010}.\footnote{Another representative approach is based on \emph{dynamic optimization}. \cite{Oyama-RSUE2009,oyama2009history} showed that, when instantaneous payoff admits a potential function, ``perfect foresight'' dynamics select the global potential maximizer if the discount rate of future payoffs is sufficiently low.}  

The stochastic stability method considers a stochastic relocation process of firms. 
Stochasticity arises because firms may make suboptimal choices. 
That is, their relocation can occur probabilistically even when the current location yields higher payoff than the alternative location. 
Such a relocation process induces a stochastic dynamics over the set of possible spatial distributions, and it will not converge to any single state. 
In such a situation, we can focus on the long-run probability distribution over the possible states. 
Under certain assumptions, the probability of a spatial distribution $\Vtx$ to occur in the long-run is shown to be proportional to $\exp\left(\eta^{-1} f(\Vtx)\right)$ where $\eta > 0$ is a parameter that governs the level of randomness (frequency of errors). 
For each fixed value of $\eta$, the state with higher potential function value is more likely to occur. 
Furthermore, when errors occur less and less often ($\eta \to 0$), the limiting behavior of the long-run probability distribution offers a method for equilibrium refinement. 
\cite{Sandholm-Book2010} (Section 12.2) shows that the set of global potential maximizers in a potential game are most likely to persist, and call them \emph{stochastically stable} states. 
Below, with this background, we will focus on global potential maximizer(s) to distill the essential implications of the HW model. 
\cref{app:stochastic_stability} provides an introductory summary of the stochastic stability approach.

In principle, the next procedure should be followed to apply the refinement based on global potential maximization.
\begin{enumerate}[label=\textbf{{Step} \arabic*}, leftmargin=4em,topsep=1em]
    \item Fix model parameters $\Vttheta\in\Theta$, where $\Theta$ is the feasible set of the parameters of interest. 
    Enumerate all spatial patterns $\Vtx^{*1}(\Vttheta)$, $\Vtx^{*2}(\Vttheta)$, $\Vtx^{*3}(\Vttheta)$, $\hdots$ that can be local maximizers of the potential function, and let $\ClE(\Vttheta) \Is \{\Vtx^{*1}(\Vttheta),\Vtx^{*2}(\Vttheta),\Vtx^{*3}(\Vttheta),\hdots\} \subset \ClX$. 
    \label{enum:step2}
    \item Given $\Vttheta$, select the global potential maximizer(s) of the potential function $f$ by the comparison of the potential values for the candidate equilibrium patterns in $\ClE(\Vttheta)$.
    \label{enum:step3}
    \item By moving $\Vttheta$ throughout $\Theta$ and repeating the two steps above, obtain the partition of $\Theta$ based on the global potential maximizer. 
    \label{enum:step4}
\end{enumerate}

The main structural parameters of interest are $\alpha$ and $\beta$. 
By definition, the set of local potential maximizers $\ClE(\Vttheta)$ contains all global potential maximizers at $\Vttheta$. 
By exhausting all possible $\Vttheta$ in the parameter space $\Theta$, we can obtain the partition of $\Theta$ based on potential maximization that provides basic insights into the implication of the model.

\section{The two-zone city}

As the simplest illustration, consider a two-zone city ($K = 2$) with symmetric transport: $t_{11} = t_{22} = 0$, $t_{12} = t_{21} = 1$. 
For convenience, we define the ease of consumers' inter-zone travel as follows:
\begin{align}
    \phi \Is \exp\left(-\beta\right).    
\end{align} 
For simplicity, we set $Q/\kappa = 1$, which is inconsequential for the relative size of equilibrium retail agglomerations in two-zone city \citep[][Theorem 6]{rijk1983equilibrium}. 

Consider the symmetric case $Q_1 = Q_2$. 
As the two zones are symmetric, \emph{uniform dispersion} of retailers $\Vtx = \left(\frac{1}{2},\frac{1}{2}\right)$ is always a spatial equilibrium. 
Also, \emph{full concentration} in either zone, $\Vtx = (1,0)$ or $(0,1)$, is always an equilibrium when $\alpha > 1$ (\cref{prop:monnopolar}). 

\begin{figure}[tb]
	\centering
	\begin{subfigure}[b]{.325\hsize}
	    \centering
	    \includegraphics{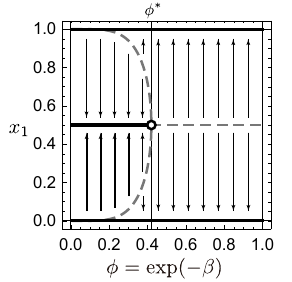}
	    \caption{Local stability} 
	    \label{fig:bif-2_ls}
	\end{subfigure}
	\hfill
	\begin{subfigure}[b]{.325\hsize}
	    \centering
	    \includegraphics{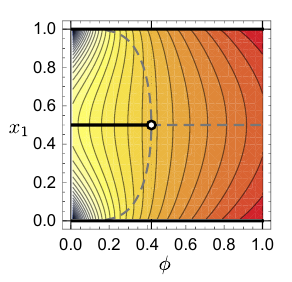}
	    \caption{Contour of $f$} 
	    \label{fig:bif-2_contour}
	\end{subfigure}
	\hfill
	\begin{subfigure}[b]{.325\hsize}
	    \centering
	    \includegraphics{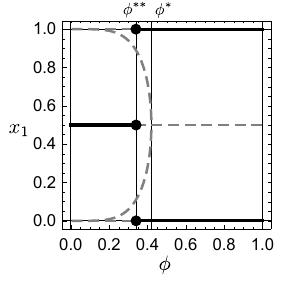}
	    \caption{Potential maximization} 
	    \label{fig:bif-2_ss}
	\end{subfigure}
	\caption{Local stability of equilibria and potential maximization in a two-zone symmetric city.}
    \FigureNote{We set $\alpha = 1.2$ and $\phi\Is\exp\left(-\beta\right)$). 
	Panel (a): The thin arrows indicate direction of adjustment under \cref{eq:D}. The solid black curves indicate locally stable equilibrium values of $x_1$. The dashed gray curves are locally unstable equilibria. 
    Panel (b): Contours of $f$ on $(\phi,x_1)$ space. 
    Red (blue) indicates a higher (lower) value of $f$. 
    Panel (c): Equilibrium refinement by global potential maximization. The black solid curves indicate potential maximizing equilibria.} 
	\label{fig:bif-2}
\end{figure} 

\subsection{Local stability approach}
\cref{fig:bif-2} shows the bifurcation diagram of spatial equilibria along the $\phi$ axis in terms of $x_1$. 
We set $\alpha = 1.2$, noting that the available empirical estimate is $\alpha = 1.18$ \citep{Ellam-etal-PRSA2018}. 
\cref{fig:bif-2_ls} considers local stability of equilibria under \cref{eq:D}. 
The black solid curves show locally stable equilibria, whereas the gray dashed curves represent locally unstable equilibria. 
Symmetric dispersion is stable for small $\phi$ and becomes locally unstable at $\phi^* = \tfrac{1 - \sqrt{\Bralpha}}{1 + \sqrt{\Bralpha}}$ with $\Bralpha = \tfrac{\alpha - 1}{\alpha}$ \citep[][Proposition~1]{Osawa-et-al-JRS2017}. 
Full concentration is locally stable for all $\phi$, in accordance with \cref{prop:monnopolar}.

\Cref{fig:bif-2_contour} shows the contours of $f$ on $(\phi,x_1)$ space. The equilibrium curves in \cref{fig:bif-2_ls} are also shown as reference. 
For each $\phi\in(0,1)$, locally stable (unstable) equilibria are local maximizers (minimizers). 
The paths of spatial equilibria trace the extrema of $f$ in the course of changing $\phi$. 

If $\phi\in(0,\phi^*)$, the local stability approach is agnostic about which equilibrium is more likely. 
However, if $\phi$ is relatively small, \cref{fig:bif-2_ls} suggests that the region of attraction for full concentration is infinitesimally small. 
While full concentration is locally stable technically, a relatively small perturbation is sufficient to nudge the equilibrium toward symmetric dispersion when $\phi$ is small. In this respect, full concentration may be ``less relevant'' if $\phi$ is small, but the local stability approach does not provide a simple means for further equilibrium refinement.

\begin{figure}[tb]
    \centering
	\begin{subfigure}[b]{.4\hsize}
	    \centering
	    \includegraphics{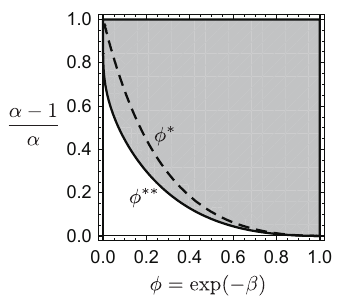}
	    \caption{Full parametric ranges ($\alpha > 1$)} 
	    \label{fig:2_ss-phi}
	\end{subfigure}
	\begin{subfigure}[b]{.4\hsize}
	    \centering
	    \includegraphics{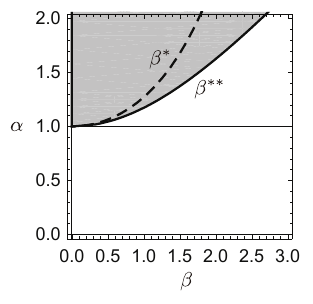}
	    \caption{Original parametric coordinates} 
	    \label{fig:2_ss-beta}
	\end{subfigure}
    \caption{Stability of equilibria in a two-zone symmetric city.} 
    \FigureNote{In Panel (a), a higher $\tfrac{\alpha - 1}{\alpha}$ corresponds to a higher $\alpha > 0$, and a higher $\phi$ corresponds to a lower $\beta$. 
    Full concentration (symmetric dispersion) is the global potential maximizer in the gray (white) region. Under the myopic dynamics \cref{eq:D}, symmetric dispersion is locally stable below the dashed curve, whereas full concentration is always locally stable. For reference, an equivalent partition of the $(\beta,\alpha)$-space is shown in Panel (b), with $\beta^* \Is - \log (\phi^*)$ and $\beta^{**} \Is - \log (\phi^{**})$, including the case $\alpha \le 1$ under which the symmetric dispersion is the unique global potential maximizer of \cref{eq:potential_maximization}.}
    \label{fig:2_ss}
\end{figure}
\subsection{Potential maximization approach}
We now apply the global potential maximization approach. 
As expected from \cref{fig:bif-2_ls}, only uniform dispersion or full concentration can maximize the potential function. 
That is, we can let $\ClE(\phi) = \left\{\left(\frac{1}{2},\frac{1}{2}\right), (1,0), (0,1)\right\}$ for all $\phi$ in \cref{enum:step2}. 
\cref{fig:bif-2_ss} is the bifurcation diagram obtained by global maximization of potential function in $\ClE(\phi)$ at each level of $\phi$ (\cref{enum:step3}) and then varying $\phi$ (\cref{enum:step4}). 
Formally, we can show the following result: 
\begin{proposition}
\label{prop:2-zone}
Consider the symmetric two-zone city where $Q_1 = Q_2 = \frac{1}{2}$. 
Then, the global potential maximizer is: symmetric dispersion if $\phi = \exp(-\beta) \in (0,\phi^{**})$, and full concentration if $\phi \in (\phi^{**}, 1)$, 
where $\phi^{**} \Is \frac{1}{2}(4^\alpha - 2 - \sqrt{4^\alpha(4^\alpha - 4)}) \in (0,1)$. 
\end{proposition}

\cref{fig:2_ss-phi} shows the partition of the parameter space based on \cref{prop:2-zone}. 
For the vertical axis, we use $\Bralpha = \tfrac{\alpha - 1}{\alpha} \in (0,1)$ to cover all $\alpha >  1$.   
The curve between the gray and white regions is $\phi^{**}$ in \cref{prop:2-zone}. 
Potential maximization selects agglomeration (dispersion) in the gray (white) region. 
Agglomeration is selected when $\phi$ is high ($\beta$ is low) and/or $\alpha$ is high. 
For reference, the dashed curve indicates the threshold $\phi^*$ mentioned earlier, above which symmetric dispersion is locally unstable. 
Symmetric dispersion is locally stable below the dashed curve, whereas agglomeration is always locally stable. \cref{fig:bif-2} corresponds to a cross section of \cref{fig:2_ss} when $\alpha = 1.2$. 
For comparison, \Cref{fig:2_ss-beta} shows the corresponding partition of the $(\beta,\alpha)$-space, where the $\alpha \le 1$ case is also shown for reference.

\subsection{Asymmetries}
With perfect symmetry as considered in \cref{prop:2-zone}, global potential maximization has no bite over local stability  for $\phi \ge \phi^*$ because both approaches choose full concentration in either zone $1$ or $2$. 
However, if the zones have asymmetric fundamentals, a sharper prediction is available: 
\begin{proposition}
\label{prop:2-zone-asym}
Consider an asymmetric two-zone city where zone $1$ is more attractive to firms than zone $2$ in terms of local demands, operating cost, or asymmetric accessibility. 
Then, full concentration in zone $2$ can never be a global potential maximizer. 
\end{proposition}

\begin{figure}[tb]
	\centering
	\begin{subfigure}[b]{.325\hsize}
	    \centering
	    \includegraphics{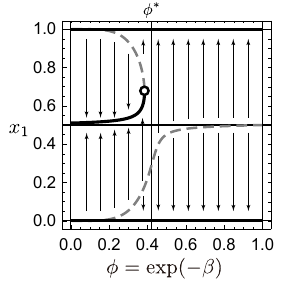}
	    \caption{Local stability} 
	    \label{fig:bif-2_asym_ls}
	\end{subfigure}
	\hfill
	\begin{subfigure}[b]{.325\hsize}
	    \centering
	    \includegraphics{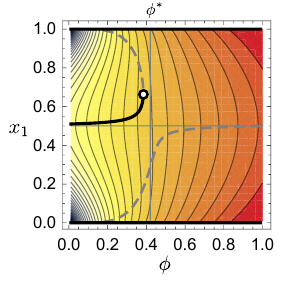}
	    \caption{Contours of $f$} 
	    \label{fig:bif-2_asym_contour}
	\end{subfigure}
	\hfill
	\begin{subfigure}[b]{.325\hsize}
	    \centering
	    \includegraphics{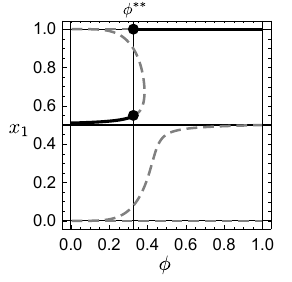}
	    \caption{Potential maximization} 
	    \label{fig:bif-2_asym_ss}
	\end{subfigure}
	\caption{Local stability of equilibria and potential maximization in a two-zone asymmetric city.}  
    \FigureNote{We set $Q_1/\kappa = 0.51 > 0.5$. 
	Panels (a)--(c) corresponds to the panels in \cref{fig:bif-2}. 
    In Panels (a) and (b), the critical value $\phi^*$ for local instability in the symmetric case is shown as a vertical line.} 
	\label{fig:bif-2_asym}
\end{figure} 
\cref{fig:bif-2_asym} shows the bifurcation diagram for an asymmetric case with a demand advantage such that $Q_1 > Q_2$. 
Other parameters are the same as \cref{fig:bif-2}. 
In \cref{fig:bif-2_asym_ls}, as in the symmetric case, full concentration in either zone is always a locally stable equilibrium. 
Reflecting the asymmetry, however, there is a locally stable interior equilibrium path such that $x_1^* > x_2^* > 0$ when $\phi$ is small, corresponding to the symmetric path for $\phi\in(0,\phi^*)$ in \cref{fig:bif-2_ls}. 
The asymmetric equilibrium is locally stable for small $\phi$ and becomes unstable earlier than the critical value $\phi^*$ in the symmetric case.\footnote{Equilibrium properties of asymmetric two-location spatial economic models are extensively studied in \cite{Berliant-Kung-RSUE2009} as well as \cite{Ikeda-etal-JRS2022}.} 
\cref{fig:bif-2_asym_ss} shows the equilibrium refinement by global potential maximization. 
In this example, there appears to be a unique $\phi^{**}\in(0,1)$ such that the global potential maximizer is: an asymmetric interior equilibrium with $x_1^* > x_2^*$ iff $\phi \in (0,\phi^{**})$, and full concentration in zone $1$ iff $\phi \in (\phi^{**}, 1)$. 
However, the analytical expression for $\phi^{**}$ is not available due to nonlinear nature of the model. 

Global maximization of the potential function allows us to focus on plausible equilibria. 
In particular, global potential maximization in asymmetric settings can yiled unambiguous equilibrium selection for almost all parametric values, as illustrated by \cref{fig:bif-2_asym_ss}.

\section{A two-dimensional city}

As a further illustration, this section provides a version of \cref{fig:2_ss} for a symmetric two-dimension geography \`a la central place theory \citep{christaller1933central,losch1940die}. 
To this end, we consider a symmetric square economy with periodic boundary conditions as shown in \cref{fig:se-unit}. 
The black points indicated sequentially numbered zones, and thin lines indicate the transportation network. 
We assume that there are $K = 8 \times 8 = 64$. 
\cref{fig:se-repeat} illustrates the periodic boundary conditions. 
For example, zone $1$ is neighboring not only to zones $2$ and $9$  
but also to zones $8$ and $57$. 
As an example of agglomeration patterns in this economy, \cref{fig:se-ex} shows a 8-centric spatial distribution. 
The gray disks schematically show the size of retail agglomeration.  
In this way, a spatial pattern in a $8\times 8$ lattice can be interpreted as an infinitely repeated pattern over a two-dimensional space. 

We set the transport cost $\ell_{ij} \ge 0$ between locations as the shortest path length between $i$ and $j$. 
For example, $\ell_{1,2} = \frac{1}{8}$, $\ell_{1,9} = \frac{3}{8}$, $\ell_{1,36} = \frac{2}{8}$, and $\ell_{1,28} = \frac{5}{8}$, where we normalize $\ell$ by $8$ so that the square has unit side length. 
Consumer demand is spatially uniform and $Q_j = \frac{Q}{K}$ for all $j\in[K]$. 
These assumptions abstract away all the exogenous advantages induced by the underlying geography, i.e., ``geographical advantage'' of \cite{matsuyama2017geographical}.  

\label{sec:two-dim} 
\begin{figure}
    \centering
    \begin{subfigure}[b]{.32\linewidth}
        \centering
        \includegraphics[width=4cm]{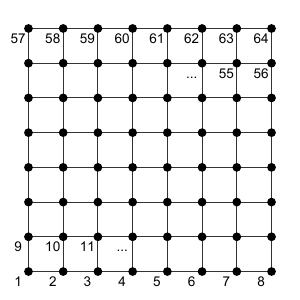}
        \caption{Square economy}
        \label{fig:se-unit}
    \end{subfigure}
    \hfill
    \begin{subfigure}[b]{.32\linewidth}
        \centering
        \includegraphics[width=4cm]{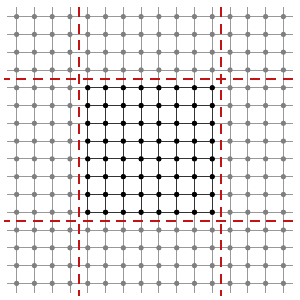}
        \caption{Periodic boundary conditions}
        \label{fig:se-repeat}
    \end{subfigure}
    \hfill
    \begin{subfigure}[b]{.32\linewidth}
        \centering
        \includegraphics[width=4cm]{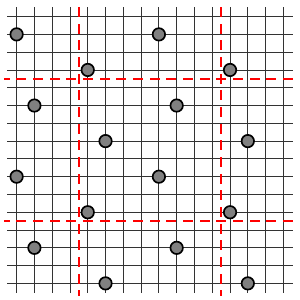}
        \caption{An agglomeration pattern}
        \label{fig:se-ex}
    \end{subfigure}
    \caption{The $8\times8$ square economy with periodic boundary conditions}
    \label{fig:square_lattice}
\end{figure}

\subsection{Invariant equilibria} 

The difficulty of considering such many-zone settings lies in \cref{enum:step2} of the procedure discussed in  \cref{sec:how-to-apply-ss}.  
The enumeration of all equilibria is practically impossible. 
For simplicity, we exclusively focus on \textit{invariant equilibria} \citep{Ikeda-etal-JEDC2018,Ikeda-etal-IJBC2019}. 
Invariant equilibria are a special class of spatial equilibrium patterns in which all retail agglomerations host the same mass of retailers: 
\begin{definition}
A spatial equilibrium $\Vtx^*\inX$ is called an \textit{invariant equilibrium} if $x_i^* = \tfrac{1}{M}$ for all $i\in\operatorname{supp}(\Vtx^*) \Is \{i\in[K]\mid x_i^* > 0\}$, where $M = \left|{\operatorname{supp}(\Vtx^*)}\right|$ is the number of retail agglomerations. 
\end{definition}

For example, $\Vtx = \left(\tfrac{1}{2},\tfrac{1}{2}\right), (1,0), (0,1)$ are the invariant equilibria in the symmetric two-zone city, and exhaust all equilibrium patterns that can be locally stable for this case. 

The procedure we follow in this section is the following: 
\begin{enumerate}[label=\textbf{{Step} \arabic*'}, leftmargin=4em,topsep=1em]
    \item Enumerate all invariant equilibria $\BrVtx^{*1}$, $\BrVtx^{*2}$, $\BrVtx^{*3}$, $\hdots$, and let $\bar{\ClE} \Is \{\BrVtx^{*1},\BrVtx^{*2},\BrVtx^{*3},\hdots\}$. 
    \label{enum:step2m}
    \item At each value of structural parameters $\Vttheta \Is \left(\alpha,\beta\right)$, check whether each invariant equilibrium in $\bar{\ClE}$ locally maximizes the potential function. Select the global maximizer(s) of potential function $f$ among locally potential-maximizing invariant equilibria. 
    \label{enum:step3m}
    \item By moving $\Vttheta$ throughout $\Theta$ and repeating \cref{enum:step3m}, obtain the partition of $\Theta$.
    \label{enum:step4m}
\end{enumerate}

For \cref{enum:step2m}, invariant equilibria in symmetric geographies can be identified using group theory \citep{Ikeda-etal-JEDC2018}. 
In a square economy with $K = n^2$ locations and periodic boundaries, it can be formally shown that $M$ must divide $8K = 8 n^2$. 
Furthermore, invariant equilibria can be enumerated by computational group theory algorithms \citep[e.g.][]{GAP4}. 
A remarkable property of invariant equilibria is that these equilibria remain to be spatial equilibria for all values of structural parameters \citep{Ikeda-etal-JEDC2018}.

\begin{figure}[tb]
    \centering
    \begin{subfigure}[b]{.49\linewidth}
        \centering
        \includegraphics[width=6.3cm]{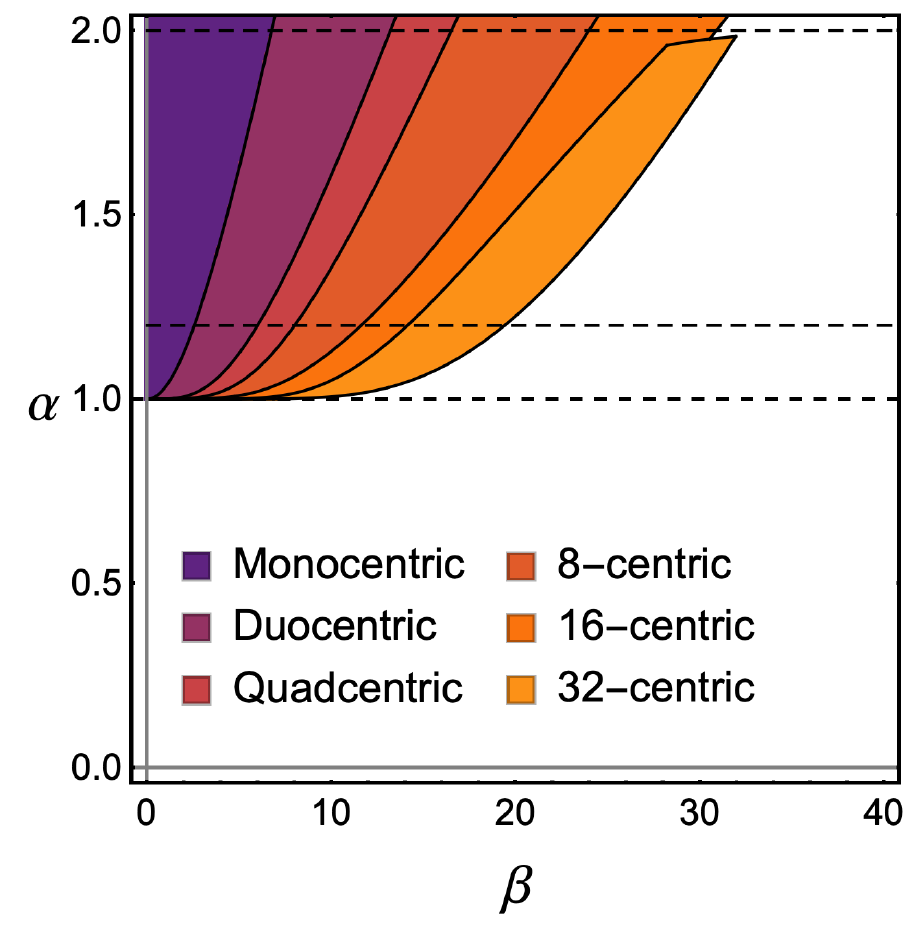}
        \vfill 
        \caption{Partition of the $(\beta,\alpha)$ space}
        \label{fig:sq_ss_partition_8} 
    \end{subfigure}
    \begin{subfigure}[b]{.49\linewidth}
        \centering
        \footnotesize
        \begin{minipage}[c]{.32\linewidth}
            \centering
            \includegraphics[width=1.5cm]{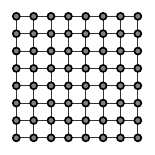}
            \parbox[c]{\linewidth}{\centering Uniform (156)}
        \end{minipage}
        \begin{minipage}[c]{.32\linewidth}
            \centering
            \includegraphics[width=1.5cm]{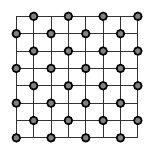}
            \parbox[c]{\linewidth}{\centering 32-centric (155)}
        \end{minipage}
        \begin{minipage}[c]{.32\linewidth}
            \centering
            \includegraphics[width=1.5cm]{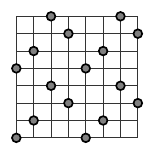} 
            \parbox[c]{\linewidth}{\centering 16-centric (149)}
        \end{minipage}
        
        \vskip 2mm
        
        \begin{minipage}[c]{.32\linewidth}
            \centering
            \includegraphics[width=1.5cm]{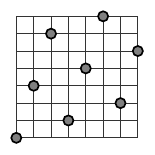}
            \parbox[c]{\linewidth}{\centering 8-centric (119)}
        \end{minipage}
        \begin{minipage}[c]{.32\linewidth}
            \centering
            \includegraphics[width=1.5cm]{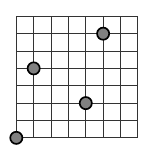} 
            \parbox[c]{\linewidth}{\centering Quadcentric (064)}
        \end{minipage}
        \begin{minipage}[c]{.32\linewidth}
            \centering
            \includegraphics[width=1.5cm]{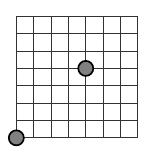}
            \parbox[c]{\linewidth}{\centering Duocentric (015)}
        \end{minipage}
        
        \vskip 2mm
        
        \begin{minipage}[c]{.32\linewidth}
            \centering
            \includegraphics[width=1.5cm]{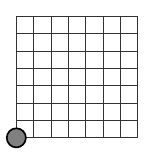}
            \parbox[c]{\linewidth}{\centering Monocentric (001)}
        \end{minipage}
        \begin{minipage}[c]{.32\linewidth}
            \ 
        \end{minipage}
        \begin{minipage}[c]{.32\linewidth}
            \ 
        \end{minipage}
        
        \vskip 1mm
        \caption{Potential-maximizing invariant patterns}
        \label{fig:sq_ss_patterns_8}
    \end{subfigure}
    \caption{Maximization of the potential function among the invariant equilibria.}
    \FigureNote{We consider a square lattice economy with $36$ locations. 
    Panel (a) shows the partition of the parameter space based on potential maximization. Panel (b) shows the associated spatial configurations. 
    The number in the label of each spatial configuration corresponds to \cref{fig:invs_sq_8} in \cref{app:invs}. 
    For the $\alpha \le 1$ case, it is known that the equilibrium is unique and globally maximizes the potential function.}
    \label{fig:sq_ss_8}
\end{figure}

A caveat is that, by definition, $\bar{\ClE}$ does not cover non-invariant equilibria in which there are retail agglomerations of different sizes. 
Therefore, $\bar{\ClE}$ does not exhaust all possible local maximizers. 
However, because there are no practical means to enumerate all non-invariant equilibria, this section resorts to the maximization of potential function over the set of invariant equilibria.

For the $8\times 8$ symmetric square economy this section considers, there are $156$ invariant equilibria. 
\cref{fig:invs_sq_8} in \cref{app:invs} shows the full list. 
All invariant equilibria exhibit geometric symmetry, and every retail agglomeration has the same market share.

\subsection{Potential maximization over invariant equilibria}
By conducting \cref{enum:step3m,enum:step4m} numerically, \cref{fig:sq_ss_8} shows the partition of the parameter space based on potentiam maximization over the set of invariant equilibria. 
Noticeably, only seven among the $156$ invariant equilibria are selected. 
Generally, retailers tend to spatially disperse if either the spatial decay parameter $\beta$ is high or agglomeration force $\alpha$ is low. 
As $\alpha$ increases or $\beta$ decreases, concentration towards a smaller number of locations occurs: the number of retail agglomerations decreases, and the spacing between them increases. 
In particular, as the distance decay rate $\beta$ goes down, we observe the ``spatial period-doubling'' behavior in which the number of retail agglomeration successively halves such that $64 \to 32 \to 16 \to 8 \to 4 \to 2 \to 1$, analogous to \cite{Osawa-et-al-JRS2017}. 
\Cref{app:tri} shows that these observations remain qualitatively valid for a different number of locations ($36 = 6\times 6$) or in a triangular grid economy.

\subsection{Local stability versus potential maximization}
Equilibrium refinement based on potential maximization is sharper than that based on local stability. 
To demonstrate this, \cref{fig:sq_ls-ss} compares the local stability approach and global potential maximization of $f$. 
Again, we limit our attention to invariant equilibria.  
We consider $\alpha = 1.2$ and $\alpha = 2.0$ as indicated in \cref{fig:sq_ss_8} by the horizontal dashed lines. 
In \cref{fig:sq_ls-ss}, the vertical axis corresponds to the $156$ invariant equilibria listed in \cref{fig:invs_sq_8}. 
The lower the index of the spatial pattern, the smaller the number of zones in which retailers locate. 
For instance, as shown in \cref{fig:sq_ss_patterns_8}, pattern $156$ corresponds to the uniform dispersion across the zones, and pattern $001$ corresponds to the full concentration in a zone. 
In \cref{fig:sq_ls-ss}, each gray solid line indicates the range of $\beta$ over which the corresponding invariant equilibrium is a local maximizer of the potential function, which is equivalent to the locally stability of the equilibrium under \cref{eq:D}. 
The black portion of each gray line, if any, indicates that the spatial configuration globally maximizes the potential function among all invariant equilibria for that range of $\beta$.

\begin{figure}[tb]
    \centering
    \begin{subfigure}[b]{.5\linewidth}
        \centering
        \includegraphics[height=6cm]{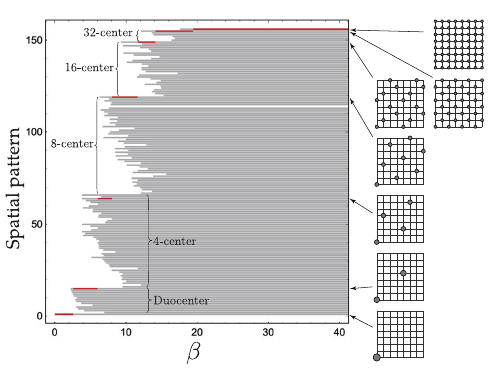}
        \caption{$\alpha = 1.2$}
        \label{fig:sq_ls-ss_12}
    \end{subfigure}
    \begin{subfigure}[b]{.49\linewidth}
        \centering
        \includegraphics[height=6cm]{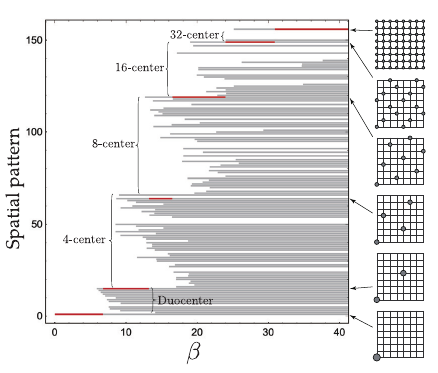}
        \caption{$\alpha = 2.0$}
        \label{fig:sq_ls-ss_20}
    \end{subfigure}
    \caption{Comparison of refinement based on local stability and poetntial maximization.}
    \FigureNote{The vertical axis corresponds to the index of the invariant patterns listed in \cref{fig:invs_sq_8} in \cref{app:invs}. 
    The gray solid lines indicate the range of $\phi$, in which the invariant equilibrium is locally stable under \cref{eq:D}. 
    The red portion on a gray solid line indicates that the equilibrium maximizes the potential value among the invariant equilibria in the range of $\phi$. 
    Panels (a) and (b) correspond to the cross-sections indicated in \cref{fig:sq_ss_partition_8}. 
    In Panel (a), seven patterns can be the global maximizer, whereas six patterns can be seen in Panel (b).
    }
    \label{fig:sq_ls-ss}
\end{figure}

\cref{fig:sq_ls-ss_12} considers the case $\alpha = 1.2$. 
The agglomeration force is relatively weak, and numerous configurations can become locally stable simultaneously. 
In particular, almost all the invariant equilibria are locally stable if $\beta$ is sufficiently large. 
Although retailers tend to agglomerate in a smaller number of locations as $\beta$ decreases, the local stability approach creates ambiguity over which configuration is the most relevant outcome. 
Instead, by considering the global maximization of the potential function, we can single out seven patterns. 

In \cref{fig:sq_ls-ss_20}, we consider the case $\alpha = 2.0$. 
Agglomeration force is so strong that retail agglomerations tend to form in smaller number of locations. 
Some invariant equilibria can never become locally stable due to the strong agglomeration force. 
However, there is still the possibility of a multiplicity of locally stable equilibria. 
Nontheless, global maximization of the potential function provides an unambiguous prediction (up to symmetry) at each level of $\beta$.  
Only six invariant equilibria can globally maximize the potential. 

Compared with the $\alpha = 1.2$ case, $32$-centric pattern is skipped when $\alpha = 2.0$ because this equilibrium is unstable for all $\beta$ for the latter case. 
This behavior is reflected in \cref{fig:sq_ss_8}, where the region for $32$-centric pattern is ``cut short'' for too high $\alpha$. 
This resembles the skipping behavior reported in \cite{Osawa-et-al-JRS2017}. 
However, such behavior may not be empirically plausible, as the available empirical estimate of $\alpha$ is around $1.18$ \citep{Ellam-etal-PRSA2018}. 

In both \cref{fig:sq_ls-ss_12,fig:sq_ls-ss_20}, the entire range of $\beta$ is covered by locally stable invariant equilibria.
It is also worth noting that in the symmetric two-zone city, asymmetric equilibria, whenever they exist, can never be locally stable \citep[][Theorem~5]{rijk1983equilibrium}.
These results suggest that, while technically possible, asymmetric equilibria may be transient patterns that connect one invariant equilibrium to another.
See \cite{Ikeda-etal-JEDC2018,Ikeda-etal-IJBC2019} for further discussion on this point in the context of a ``new economic geography'' model.
Although the HW framework assumes no congestion effects within each zone, if such within-location congestion is significant \citep[as in, e.g.,][]{Helpman-Book1998,Allen-Arkolakis-QJE2014} then non-invariant equilibria may become more relevant.

\section{Concluding remarks}
\label{sec:conclusion}

The Harris and Wilson model is a parsimonious framework for the formation of urban spatial structures. 
This study introduces a new approach for equilibrium refinement based on potential game theory.  
We first observed that the model is a large-population potential game and use (global) potential maximization as a refinement of equilibrium. 
For the asymmetric two-zone setting, the global maximization of a potential function allows unambiguous prediction of the equilibrium spatial structure unlike the local stability approach. 
Our results corroborates with other previous numerical observations in the literature that lowering transport costs promotes concentration of retailers toward smaller number of locations. 

To enumerate candidate spatial configurations in two-dimensional settings, we employed a systematic approach developed by \cite{Ikeda-etal-JEDC2018,Ikeda-etal-IJBC2019} to consider invariant equilibria, the set of equilibria that features geographical symmetry. 
A similar approach to focus on symmetric patterns was adopted by \cite{Osawa-Akamatsu-2020} in the context of an urban economics model. 
The limitation of this approach is that, by construction, asymmetric spatial configurations are abstracted away. 
Further research is needed on this respect. 

The simplicity of the Harris and Wilson framework allows its application in diverse contexts. 
The existence of a potential function enables researchers to develop a unified framework of parameter estimation \citep{Ellam-etal-PRSA2018}. 
Some studies aim to deepen the physics of \cite{Wilson-RSI2007}'s Boltzmann--Lotka--Volterra framework. For example, \cite{Crosato-etal-RSOS2018} considered the thermodynamic efficiency of urban transformation. 
\cite{Slavko-etal-PRE2019} pointed out that an important generalization is to consider the resettlement of consumers in considering the long-run evolution of urban spatial structure. 
This means that there are two types of qualitatively different actors in the model, in contrast to the original HW framework. 
\cite{Osawa-Akamatsu-2020} showed that the potential maximization approach employed in this study could be an effective method of analysis for models with multiple types of agents. 
Further development of the potential game approach for modeling urban spatial structure will be important for both theory and applications.

\appendix
\crefalias{section}{appendix}

\section{Proofs}
\label{app:proofs} 

\begin{proof}[Proof of \cref{prop:monnopolar}]
    If $x_i = Q/\kappa_i$ and $x_j = 0$ ($j\ne i$), then it is a strict Nash equilibrium because $\pi_i(\Vtx) = 0 > -\kappa_j = \pi_j(\Vtx)$ ($j\ne i$) if $\alpha > 1$. It is known that strict equilibria are locally stable under a wide range of dynamics, including \cref{eq:D} \citep{sandholm2014local}. 
\end{proof}

\begin{proof}[Proof of \cref{prop:2-zone}]
    Let $\kappa = Q = 1$ to economise notations, which is inconsequential as $Q/\kappa = 1$. 
    We first list all spatial equilibria. 
    The symmetric equilibrium $\BrVtx = (\frac{1}{2},\frac{1}{2})$ is always a spatial equilibrium. 
    For $\alpha > 1$, $\Vtx^{*1} = (1,0)$ and $\Vtx^{*2} = (0,1)$ are always spatial equilibria. 
    Theorem 3 in \cite{rijk1983uniqueness} shows the following: 
    when $1 < \alpha < \alpha^* \Is \frac{(1 + \phi)^2}{4\phi}$, or equivalently, when $0 < \phi < \phi^* \Is (\sqrt{\alpha} - \sqrt{\alpha - 1})^2$, there are exactly three interior equilibria; 
    other than $\BrVtx$, there are asymmetric equilibria of the form $\Vtx^* = \BrVtx \pm t \cdot (1 , -1)$ with some $t \in (0,\frac{1}{2})$; 
    for $\alpha \ge \alpha^*$, there is a unique interior equilibrium, which is $\BrVtx$.
    From Theorem 5 of \cite{rijk1983equilibrium}, the asymmetric equilibria are always locally unstable, meaning that they are local minima of the potential function. 
    
    That said, only patterns in $\ClE \Is \{\Vtx, \Vtx^{*1}, \Vtx^{*2}\}$ can be local maximizers of $f$. Let $\Vtx(t) \Is (\frac{1}{2} + t, \frac{1}{2} - t)$. 
    With abuse of notation, let $f(t) \Is f(\Vtx(t))$. Then, 
    \begin{align}
        f(t) = 
        \frac{1}{2\alpha}\log 
        \left[\left(\left( \tfrac{1}{2} + t \right)^\alpha + \phi \left( \tfrac{1}{2} - t \right)^\alpha\right) 
        \left(\left( \tfrac{1}{2} - t \right)^\alpha + \phi \left( \tfrac{1}{2} + t \right)^\alpha\right) 
        \right]
        - 1. 
    \end{align} 
    From symmetry, we focus on $t \in [0,\frac{1}{2}]$. 
    We have $f'(\frac{1}{2}) = 1$, showing that $\Vtx^{*1} = \Vtx(\frac{1}{2})$ is a local maximizer of $f$ for any $\phi$ and $\alpha$. 
    Also, $f'(0) = 0$, meaning that $\BrVtx = \Vtx(0)$ is an extremum of $f$. We have
        $f''(0) = 4\alpha \left(\Bralpha - \chi^2\right)$ 
    with $\Bralpha \Is \tfrac{\alpha - 1}{\alpha}$ and $\chi \Is \tfrac{1 - \phi}{1 + \phi}$. The condition $f''(0) < 0$ holds true iff $\phi < \phi^*$, showing $\BrVtx$ is a local maximizer iff $\phi \in (0,\phi^*)$.  

    Let $\Delta(\phi) \Is f(\frac{1}{2}) - f(0) = \frac{1}{2\alpha} \log \frac{2^{2\alpha}\phi}{(1 + \phi)^2}$. 
    For any $\alpha > 1$, $\phi = \phi^{**} \Is \frac{1}{2}(4^\alpha - 2 - \sqrt{4^\alpha(4^\alpha - 4)}) \in (0,1)$ solves the equation $\Delta(\phi) = 0$. 
    We confirm $\Delta'(\phi) = \frac{1 - \phi}{2\alpha\phi(1 + \phi)} > 0$ for all $\phi\in(0,1)$, showing that $f(\frac{1}{2}) > f(0)$ for all $\phi \in (\phi^{**}, 1)$ and $f(\frac{1}{2}) < f(0)$ for all $\phi \in (0,\phi^{**})$. 
    Finally, we confirm $\phi^{**} < \phi^*$ for all $\alpha > 1$, showing that $\BrVtx$ is indeed a local maximizer when $\phi\in(0,\phi^{**})$. 
\end{proof}

\begin{proof}[Proof of \cref{prop:2-zone-asym}]
We consider three important types of exogenous differences. 

\emph{Difference in $Q_k$}. For $\alpha > 1$, $\Vtx^{*1} = (1,0)$ and $\Vtx^{*2} = (0,1)$ are always (locally stable) spatial equilibria even under the asymmetry $Q_1 = \rho \kappa$ and $Q_2 = (1 - \rho)\kappa$ with $\rho > 1/2$. 
However, for the potential values at $\Vtx^{*1}$ and $\Vtx^{*2}$, $f(\Vtx^{*1}) - f(\Vtx^{*2}) = - \kappa(2\rho - 1)\log(\phi) / \alpha > 0$ for any $\phi \in (0,1)$, $\alpha > 1$, and $\rho > 1/2$, indicating that $\Vtx^{*2}$ can never be the global maximizer of $f$. 

\emph{Difference in $\kappa_{i}$}. 
In this case, $\Vtx^{*1} = (Q/\kappa_1, 0)$ and $\Vtx^{*1} = (0, Q/\kappa_2)$. 
As discussed in the main text, the total equilibrium cost always satisfies $\sum_{i} \kappa_i x_i = Q$. 
Thus, only accessibility differences affect potential values. 
We can compute $f(\Vtx^{*1}) - f(\Vtx^{*2}) = 2Q \log (\kappa_2/\kappa_1)$, showing that $f(\Vtx^{*1}) > f(\Vtx^{*2})$ iff $\kappa_1 < \kappa_2$, i.e., if zone $1$ has a cost advantage.

\emph{Difference in transport costs}. Let $\phi_{ij} \Is \exp(-\beta \ell_{ij})$. Suppose $Q_1 = Q_2 = \kappa/2$ and that $\phi_{21} = \rho\phi > \phi = \phi_{12}$ with $\rho > 1$, so that consumers in zone $2$ have better access to retailers in zone $1$ than consumers in zone $1$ has to retailers in zone $2$. 
In other words, firms in zone $2$ are at a disadvantage because consumers from zone $2$ can access zone $1$ more easily than vice versa. Such directional asymmetries in accessibility are common in urban settings—for example, due to one-way streets, time-dependent congestion patterns, or public transit routes with asymmetric coverage or frequency. Geographic features like rivers or hills may also induce such imbalances, affecting firms’ effective market reach and local competition. 
We can verify that $f(\Vtx^{*1}) - f(\Vtx^{*2}) = (Q/\alpha) \log (\rho) > 0$ since $\rho > 1$. 
\end{proof}

\section{Stochastic stability in potential games} 
\label{app:stochastic_stability}

\cite{Sandholm-Book2010}, Sections 11.5 and 12.2 develop a theory under which the \textit{global maximizers} of the potential function are shown to be ``stochastically stable.'' 
Below, we review the essence of his analysis in an accessible manner. 
Our presentation is inevitably brief. For a complete and rigorous treatment we refer to the original text. 
See also \cite{wallace2015stochastic} for a broader survey on stochastic stability approaches in game theory. 

To define the stochastic stability of a state, a stochastic relocation dynamics of retailers must be introduced. 
In doing so, we regard the model in \cref{sec:model} as a continuous (or large-population) and deterministic limit of a discrete and stochastic version of the model. 

Suppose there is a \textit{finite} (but large) number of retailers and let $N\in\BbZ$ be the number of retailers. 
Then, a spatial distribution of finite retailers can be seen as an element of the discrete set $\ClX^N$ defined by 
    $\ClX^N 
    \Is 
    \left\{
        \Vtx 
        \inX
    \Mid 
        N \Vtx \in\BbZ^K
    \right\}$. 
For $\Vtx\inX^N$, we have $x_i \in \{0,\frac{1}{N}, \frac{2}{N}, \hdots , \frac{N - 1}{N}, 1\}$. 
For brevity, we further assume that $\kappa_i = \kappa\ \forall i$ and $Q / \kappa = 1$ so that $\sum_{i\in[K]} x_i = 1$ in the continuous case, thereby $\ClX^N \in \ClX$. 
For simplicity, we abstract from the outside option $i = 0$ because any equilibrium in the continuous case are in $\ClX$ (\cref{lem:X} on \cpageref{lem:X}). 

Every retailer receives action revision opportunities according to a Poisson process with a unit rate. 
When a retailer in zone $i$ receives a revision opportunity at state $\Vtx\inX^N$, it switches from zone $i$ to $j$ according to the logit rule: 
\begin{align}
    \rho_{ij}^N(\Vtx) = \dfrac{\exp\left(\eta^{-1} \pi_{ij}^N(\Vtx)\right)}{\sum_{k\in[K]} \exp\left(\eta^{-1} \pi_{ik}^N(\Vtx)\right)}, 
    \label{eq:logit_rule}
\end{align}
where $\eta > 0$. 
Here, we suppose that firms in the finite-agent game are ``clever'' in the sense that, upon their choice, they evaluate hypothetical payoff \emph{after} their unilateral move \citep[][Section 11.4.2]{Sandholm-Book2010}. That is, $\pi_{ij}^N(\cdot)$ in \cref{eq:logit_rule} is defined as follows: 
\begin{align}
    \pi_{ij}^N(\Vtx) = \pi_j\left(\Vtx + \tfrac{1}{N} (\Vte_j - \Vte_i)\right), 
\end{align}
where $\pi_j(\cdot)$ is the payoff function for the HW model; $\Vte_i$ is the $i$th standard basis in $\BbR^K$, so that $\tfrac{1}{N} (\Vte_j - \Vte_i)$ is the displacement from the current state $\Vtx$ when a firm moves from $i$ to $j$. 
This rule is an instance of ``direct exponential protocols'' (Ibid., Section 11.5.2). 

We have $\rho_{ij}(\Vtx) > \rho_{ik} (\Vtx)$ if $\pi_{ij}^N(\Vtx) > \pi_{ik}^N(\Vtx)$, meaning that retailers prefer locations with higher profit. 
Note that, however, the probability of switching to a lower-profit zone is not zero. 
The parameter $\eta$ can be interpreted as the level of noise in retailers' choice. 
When $\eta \to 0$, every retailer switches to $j$ with the highest profit with probability $1$. 
If $\eta$ is high instead, retailers may choose less profitable zones than their current choice.

These assumptions induce a stochastic dynamic for retailers' spatial distribution.
It is a Markov process $\{\VtX_t^N\}$ on the discrete state space $\ClX^N$ with a jump rate $N$, and the transition probabilities from state $\Vtx\inX^N$ to $\Vty\inX^N$ are as follows. 
\begin{align}
    P_{\Vtx\to\Vty}^N
    = 
    \begin{cases}
        x_i \rho_{ij}(\Vtx) & \text{if } \Vty = \Vtx + \tfrac{1}{N}(\Vte_j - \Vte_i),\ j\ne i
        \\
        \sum_{i\in[K]} x_i \rho_{ii}(\Vtx) & \text{if } \Vty = \Vtx, 
        \\
        0 &\text{otherwise}. 
    \end{cases}
\end{align}
Under this stochastic evolutionary law, the state can move only to neighboring states in $\ClX^N$. 

For each given $N$ and $\eta$, the Markov process $\{\VtX_t^N\}$ admits a unique stationary distribution $\mu^{N,\eta}$ on $\ClX^N$ as follows (Ibid., Theorem 11.5.12):  
\begin{align}
    \mu^{N,\eta}(\Vtx) = \frac{1}{Z} \frac{N!}{\prod_{k\in[K]} (N x_k)!} \exp \left(\eta^{-1} f^N(\Vtx) \right), 
    \label{eq:stationary_distribution}
\end{align}
where $Z > 0$ is the normalizing constant to ensure $\sum_{\Vtx\inX^N} \mu^{N,\eta}(\Vtx) = 1$. 
The function $f^N:\ClX^N \to \BbR$ is a discrete analog for the potential function $f$ for the finite-population case (Ibid., Section 11.5.1), and $\{\tfrac{1}{N} f^N\}$ converges uniformly to $f$ as $N \to \infty$.\footnote{Since the main text focuses on potential maximization in the continuous model, we do not explicitly introduce $f^N$ here. The finite-population potential function is defined to satisfy $f^N(\Vtx) - f^N(\Vtx -\frac{1}{N} \Vte_i) = \pi_i(\Vtx)$ for all $\Vtx\in\ClX^N$ and $i\in[K]$. Such a function can be constructed by a discrete approximation of the coordinate-wise discretized line integral of the continuous payoff function $\Vtpi(\Vtx)$ along the line segment from $(0,0,\hdots,0)$ to $\Vtx$.}

In evolutionary game theory, a state is said to be \textit{stochastically stable} when the stationary distribution of a stochastic dynamic assigns a positive weight on the state in some limits of the structural parameters of the dynamic adjustment process. 
The simplest example is \textit{stochastic stability in the small noise limit}. 
A state $\Vtx^* \in \ClX^N$ is stochastically stable in the small noise limit when 
\begin{align}
    \lim_{\eta \to 0} \mu^{N,\eta} (\Vtx) > 0. 
\end{align}
In the limit $\eta \to 0$, retailers choose zones with higher profit with higher probability. 
The small noise limit is a deterministic limit where noise vanishes and retailers recover optimal choice behavior. 

Small noise limit can be understood with the formula \cref{eq:stationary_distribution}. 
We have
\begin{align}
    \frac {\mu^{N,\eta}(\Vtx)}{\mu^{N,\eta}(\Vty) } 
    = 
    \underbrace{\frac{\prod_{k\in[K]} (N y_k)!}{\prod_{k\in[K]} (N x_k)!}}_{\text{constant in $\eta$}. }
    \exp\left(
    \eta^{-1}
    \left(
        f^N(\Vtx) - f^N(\Vty)
    \right) 
    \right)
\end{align}
for two states $\Vtx,\Vty\inX^N$. 
If $f^N(\Vtx) - f^N(\Vty) > 0$, then the right-hand side grows infinitely large as $\eta \to 0$. 
That is, $\mu^{N,\eta}$ assigns higher and higher probability on the states with larger values of $f^N$ when $\eta$ goes smaller and smaller. 
In the limit, $\mu^{N,\eta}$ concentrates on the states that \textit{globally maximize} $f^N$. 
The global maximizers of $f^N$ are stochastically stable in the small noise limit under a fixed $N$. 

In a similar spirit to the small noise limit, the \textit{double limits} considers a situation where both $N \to \infty$ and $\eta \to 0$. 
By taking these two limits, we recover the large-population game as laid out in \cref{sec:model}, in which retailers do not incur errors, and the set of retailers is a continuum. 
Thus, stochastic stability in double limits provides a refinement procedure for the deterministic large-population model. 

\cite{Sandholm-Book2010} (Corollary 12.2.5) establishes a stochastic stability result for the double limits in potential games under the logit choice rule \cref{eq:logit_rule}. Specifically, it shows that 
\begin{align}
    & \lim_{N \to \infty} \lim_{\eta \to 0} \max_{\Vtx\in\ClX^N}\left| \frac{\eta}{N} \log \mu^{N,\eta}(\Vtx) - \Delta f(\Vtx) \right| = 0 \text{\quad and}\\
    & \lim_{\eta \to 0} \lim_{N \to \infty} \max_{\Vtx\in\ClX^N}\left| \frac{\eta}{N} \log \mu^{N,\eta}(\Vtx) - \Delta f(\Vtx) \right| = 0, 
\end{align}
where $\Delta f(\Vtx) \Is f(\Vtx) - \max_{\Vty\in\ClX} f(\Vty)$ is a translated version of the potential function for the continuous model. 
By definition, we have $\Delta f(x) \le 0$ with equality only at the global maximizers of $f$. 
The characterization is called stochastic stability in the double limits in the \emph{weak} sense (Ibid., Section 12.1.3).
It means that
\begin{align}
    \mu^{N,\eta}(\Vtx) = \exp\left(N \eta^{-1} \Delta f(\Vtx) + o(1)\right)
\end{align}
where $o(1)$ is a term that goes to zero uniformly as $N \to \infty$ and/or $\eta \to 0$. 
Thus, $\Delta f(\Vtx)$ can be seen as the exponential \emph{decay rate} of the probability mass on $\Vtx$. 
For a state such that $\Delta f(\Vtx) < 0$, the probability mass on it must vanish at an exponential rate.  
Then, the stationary distribution must concentrate on the set of global maximizers of the potential function $f$ as $N \to \infty$ and $\eta \to 0$. 
That is, the set of global potential maximizers in a large-population potential game is (weakly) stochastically stable in the double limits.

\section{A modified potential function} 
\label{app:ellam}

The potential function considered in \cite{Ellam-etal-PRSA2018} is the following (their equation (2.9)): 
\begin{align}
    g(\Vtm) 
    = 
    \frac{1}{\alpha} \sum_{j\in[K]} 
    Q_j \log \sum_{i\in[K]} 
    \exp\left(\alpha m_i - \beta t_{ij} \right)
    - 
    \kappa 
    \sum_{j\in[K]} 
    \exp(m_j)
    +
    \delta
    \sum_{j\in[K]} m_j, 
\end{align}
where $\Vtm$ is the log size of retail agglomeration (i.e., $m_i = \log x_i$), $\kappa_i = \kappa$ for all $i\in[K]$, and $\delta > 0$ is a parameter. 
If we rewrite $g$ as a function of $\Vtx$, we see
\begin{align}
    g(\Vtx) 
    = 
    \underbrace{\frac{1}{\alpha} \sum_{j\in[K]} 
    Q_j \log \sum_{i\in[K]} 
    x_i^\alpha \exp\left( - \beta t_{ij} \right)
    - 
    \kappa 
    \sum_{j\in[K]} 
    x_j}_{f(\Vtx) \text{ in \cref{eq:HW-P}}}
    +
    \underbrace{\delta
    \sum_{j\in[K]} \log x_j}_{\text{additional term}}. 
\end{align} 
The potential function is motivated by the original HW framework in \cref{eq:HW-P}, however, the additional term is introduced so that their main stochastic differential equation model has a well-defined stationary distribution. 
As \cite{Ellam-etal-PRSA2018} discusses, the additional term prevents zones from ``collapsing,'' that is, for some $x_i$ to become zero. 

The first interpretation is based on large-population potential games. 
As we have seen, 
\begin{align}
    g(\Vtx) = f(\Vtx) + \delta \sum_{j\in[K]} \log x_j 
\end{align}
where $f$ is the potential function for the original HW model, defined in \cref{eq:HW-P}. 
In fact, if we interpret the modified potential function as the integral of the underlying profit function $\TlVtpi(\Vtx)$, we have 
\begin{align}
    \Tlpi_i(\Vtx) = \PDF{g(\Vtx)}{x_i} = \pi_i(\Vtx) + \frac{\delta}{x_i}. 
\end{align}
Since $\pi_i(\Vtx) \to - \kappa$ when $x_i \to 0$, we see that $\Tlpi_i(\Vtx) \to \infty$ when $x_i \to 0$. 
Thus, at any spatial equilibrium (\cref{def:equilibrium}), every zone should have retailers. 
The additional term $\delta/x_i$ can be interpreted to represent some congestion force. 
Economic foundations for such a term may be land input in firms' production \citep[cf.][in the context of land consumption of households]{Picard-Tabuchi-JET2013}. 
Similar properties can emerge when we consider idiosyncratic preference of consumers in spatial models \citep{Behrens-Murata-JUE2021}.

The second interpretation is an approximation of the constrained maximization problem \cref{eq:potential_maximization} associated with the original HW model. 
If we introduce the ``log barrier'' term $\delta \log x_j$ corresponding to each nonnegativity constraint $x_j \ge 0$, the maximization problem \cref{eq:potential_maximization} becomes
\begin{align}
    \max_{\Vtx} g(\Vtx) \quad \text{s.t.}\quad \sum_{i\in[K]} \kappa_i x_i = Q, 
\end{align}
where the nonnegativity constraints are approximated by the additional term. 

\section{Triangular grid economy}
\label{app:tri}

This appendix compares the square economy with another two-dimensional space, a symmetric triangular grid economy. 
A triangular grid economy with periodic boundaries is important because the hexagonal market area envisaged by central place theory \citep{christaller1933central,losch1940die} can endogenously emerge in this setting \citep{ikeda2014bifurcation}. 
In the context of the HW model, \cite{beaumont1981changing} provided a numerical investigation on a hexagonal economy with a triangular grid. 

To lessen computational burdens, we on a triangular grid with $6\times 6$ locations.  
There are $65$ invariant equilibria, which we list in \cref{app:invs}. 

\Cref{fig:hx_ss} show the partition of the parameter space based on potential maximization. 
To show the entire parametric range, we take the parametrization $\phi \Is \exp(-\beta / 6)$, which shows the whole $\beta\in(0,\infty)$. 
Also, vertical axis is chosen to $\frac{\alpha - 1}{\alpha}$ so that all $\alpha > 1$ can be shown as the $(0,1)$ interval. 
Analogous to \cref{fig:sq_ss_8}, the selected spatial configurations are aligned from the bottom left to the top right according to the decreasing order in terms of the number of retail agglomerations. 

In \cref{fig:hx_ss}, in addition to the uniform distribution and full agglomeration in a zone, there are two representative configurations that occupy relatively large regions in the parameter space: 12-centric and tricentric patterns. 
The former corresponds to Christaller's $k = 3$ system, as $\frac{36}{12} = 3$. Both the patterns feature hexagonal market area considered in central place theory. 
As we compare \cref{fig:hx_ss} and \cref{fig:sq_ss_8}, we observe that the basic implication is robust irrespective of the underying geography. 

\begin{figure}[tb]
    \centering

    \centering
    \begin{subfigure}[b]{.49\linewidth}
        \centering
        \includegraphics{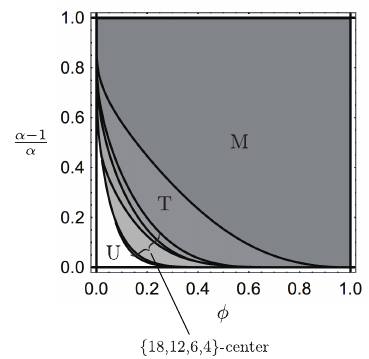}
        \vfill 
        \caption{Partition of the parameter space}
    \end{subfigure}
    \begin{subfigure}[b]{.49\linewidth}
        \centering
        \footnotesize
        \begin{minipage}[c]{.32\linewidth}
            \centering
            \includegraphics[width=2.55cm]{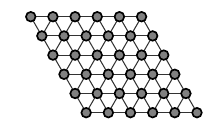}
            \parbox[c]{\linewidth}{\centering Uniform (65)}
        \end{minipage}
        \begin{minipage}[c]{.32\linewidth}
            \centering
            \includegraphics[width=2.55cm]{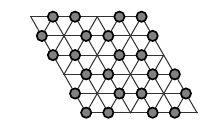}
            \parbox[c]{\linewidth}{\centering 18-centric (63)}
        \end{minipage}
        \begin{minipage}[c]{.32\linewidth}
            \centering
            \includegraphics[width=2.55cm]{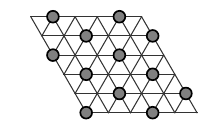} 
            \parbox[c]{\linewidth}{\centering 12-centric (57)}
        \end{minipage}
        
        \vskip 2mm
        
        \begin{minipage}[c]{.32\linewidth}
            \centering
            \includegraphics[width=2.55cm]{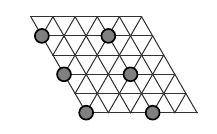}
            \parbox[c]{\linewidth}{\centering 6-centric (33)}
        \end{minipage}
        \begin{minipage}[c]{.32\linewidth}
            \centering
            \includegraphics[width=2.55cm]{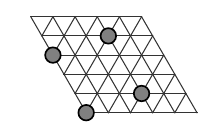}
            \parbox[c]{\linewidth}{\centering Quadcentric (25)}
        \end{minipage}
        \begin{minipage}[c]{.32\linewidth}
            \centering
            \includegraphics[width=2.55cm]{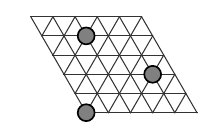} 
            \parbox[c]{\linewidth}{\centering Tricentric (14)}
        \end{minipage}
        
        \vskip 2mm
        
        \begin{minipage}[c]{.32\linewidth}
            \centering
            \includegraphics[width=2.55cm]{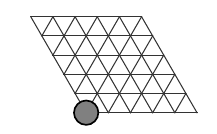}
            \parbox[c]{\linewidth}{\centering Monocentric (01)}
        \end{minipage}
        \begin{minipage}[c]{.32\linewidth}
            \ 
        \end{minipage}
        \begin{minipage}[c]{.32\linewidth}
            \ 
        \end{minipage}
        
        \vskip 1mm
        \caption{Potential-maximizing patterns}
    \end{subfigure}
    \caption{Maximization of the potential function among the invariant equilibria (triangular lattice).}
    \FigureNote{We consider a triangular lattice economy with $36$ locations. Panel (a) shows the partition of the parameter space based on potential maximization among the invariant equilibria. 
    Panel (b) shows the associated spatial configurations. 
    The number in the label of each spatial configuration corresponds to \cref{fig:invs_hx} in \cref{app:invs}. 
    In Panel (a), the letters M and T indicate the mono- and duo-centric equilibria, respectively; U indicates uniform equilibrium; 
    the $\{18,12,6,4\}$-centric equilibrium patterns are sequentially aligned from left to right on the $\phi$ axis.}
    \label{fig:hx_ss}
\end{figure}

\section{Invariant equilibria}
\label{app:invs}

\Cref{fig:invs_sq_8} lists all the invariant equilibria for the $8\times 8$ square economy with periodic boundaries. 
\cref{fig:se-ex} in the main text shows pattern 64 in \cref{fig:invs_sq_8}. 
\cref{fig:invs_hx} shows all the invariant equilibria for the $6\times 6$ triangular tird economy with periodic boundaries. 
We observe that retail agglomerations are symmetrically placed over the geography. 
These invariant equilibria are characterized by the \emph{group} $G$ that represents the symmetry of the economy. 
For example, the $8\times 8$ square geography is invariant (symmetric) under $90^\circ$, $180^\circ$, and $270^\circ$ rotation as well as horizontal and vertical translation. 
The group $G$ is a mathematical object that encapsulates such symmetry.
By exploiting this symmetry, the GAP software \citep{GAP4} was employed to enumerate the invariant equilibria for the square and triangular grid economy considered in this study. 
For each geography, we first enumerate all the \emph{subgroups} $\{G'\}$ in the ground group $G$.  
Subsequently, we apply \emph{orbit decomposition} for each subgroup $G'$, which is a partitioning of the set of locations $[K]$ into \emph{equivalence class} defined by the \emph{action} of $G'$ (the permutations of zone indices induced by $G'$). 
The support $\operatornamewithlimits{supp}(\BrVtx^*)$ of each invariant equilibria $\BrVtx^*$ corresponds to one of the partitioned components of $[K]$. 
See \cite{Ikeda-etal-JEDC2018} and \cite{Ikeda-etal-IJBC2019} for group-theoretic foundations of invariant equilibria in square and triangular geographies, respectively, with an arbitrary number of locations.

\begin{figure}
\centering
\begin{minipage}[tb]{.95\linewidth}
\centering
\includegraphics[width=1.25cm]{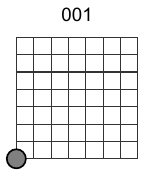}
\includegraphics[width=1.25cm]{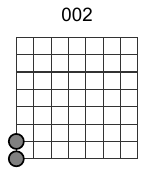}
\includegraphics[width=1.25cm]{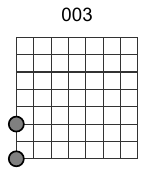}
\includegraphics[width=1.25cm]{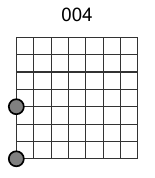}
\includegraphics[width=1.25cm]{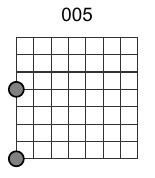}
\includegraphics[width=1.25cm]{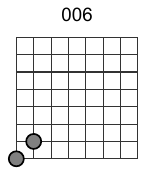}
\includegraphics[width=1.25cm]{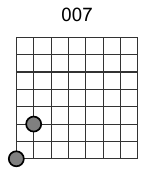}
\includegraphics[width=1.25cm]{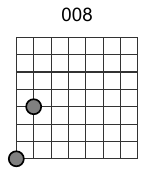}
\includegraphics[width=1.25cm]{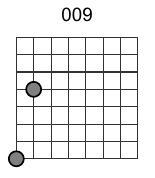}
\includegraphics[width=1.25cm]{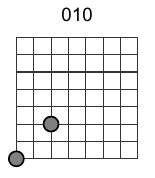}

\includegraphics[width=1.25cm]{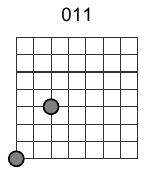}
\includegraphics[width=1.25cm]{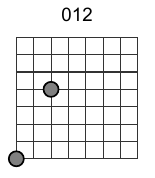}
\includegraphics[width=1.25cm]{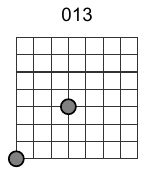}
\includegraphics[width=1.25cm]{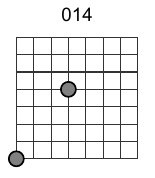}
\includegraphics[width=1.25cm]{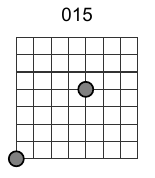}
\includegraphics[width=1.25cm]{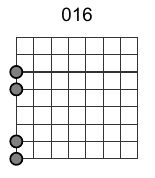}
\includegraphics[width=1.25cm]{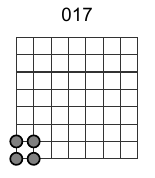}
\includegraphics[width=1.25cm]{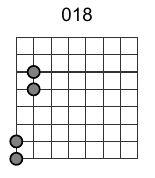}
\includegraphics[width=1.25cm]{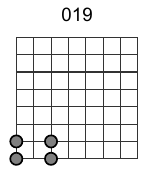}
\includegraphics[width=1.25cm]{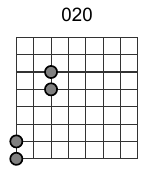}

\includegraphics[width=1.25cm]{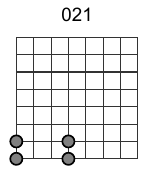}
\includegraphics[width=1.25cm]{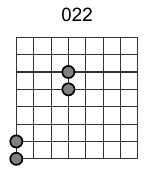}
\includegraphics[width=1.25cm]{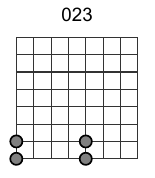}
\includegraphics[width=1.25cm]{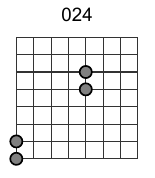}
\includegraphics[width=1.25cm]{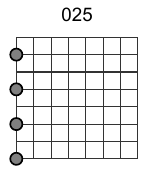}
\includegraphics[width=1.25cm]{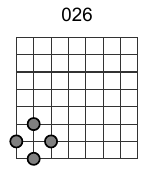}
\includegraphics[width=1.25cm]{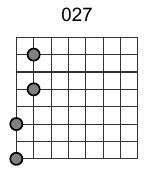}
\includegraphics[width=1.25cm]{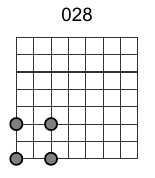}
\includegraphics[width=1.25cm]{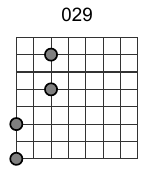}
\includegraphics[width=1.25cm]{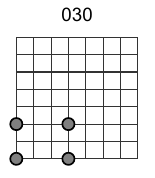}

\includegraphics[width=1.25cm]{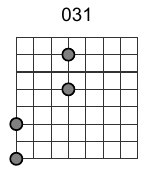}
\includegraphics[width=1.25cm]{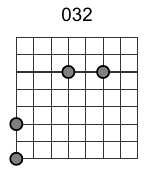}
\includegraphics[width=1.25cm]{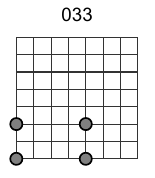}
\includegraphics[width=1.25cm]{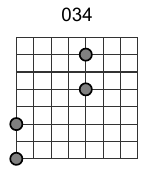}
\includegraphics[width=1.25cm]{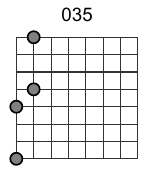}
\includegraphics[width=1.25cm]{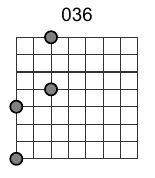}
\includegraphics[width=1.25cm]{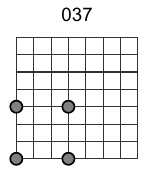}
\includegraphics[width=1.25cm]{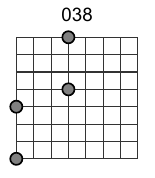}
\includegraphics[width=1.25cm]{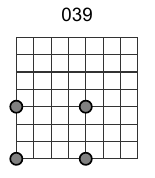}
\includegraphics[width=1.25cm]{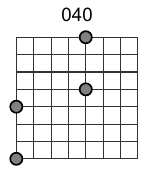}

\includegraphics[width=1.25cm]{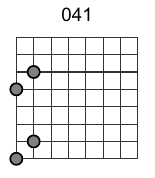}
\includegraphics[width=1.25cm]{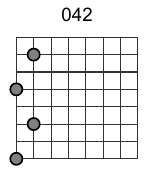}
\includegraphics[width=1.25cm]{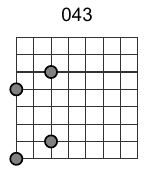}
\includegraphics[width=1.25cm]{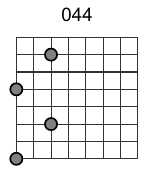}
\includegraphics[width=1.25cm]{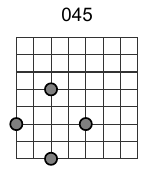}
\includegraphics[width=1.25cm]{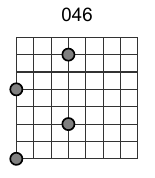}
\includegraphics[width=1.25cm]{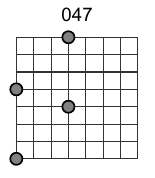}
\includegraphics[width=1.25cm]{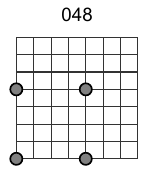}
\includegraphics[width=1.25cm]{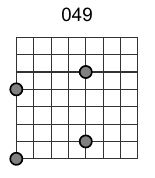}
\includegraphics[width=1.25cm]{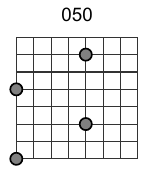}

\includegraphics[width=1.25cm]{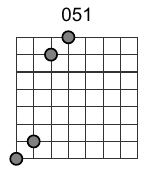}
\includegraphics[width=1.25cm]{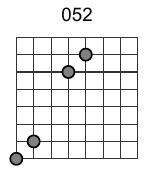}
\includegraphics[width=1.25cm]{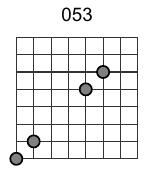}
\includegraphics[width=1.25cm]{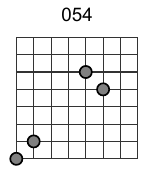}
\includegraphics[width=1.25cm]{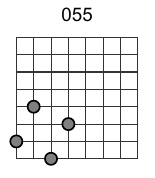}
\includegraphics[width=1.25cm]{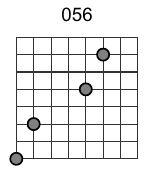}
\includegraphics[width=1.25cm]{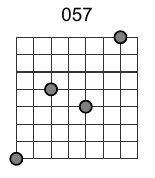}
\includegraphics[width=1.25cm]{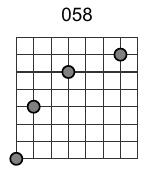}
\includegraphics[width=1.25cm]{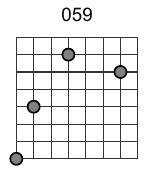}
\includegraphics[width=1.25cm]{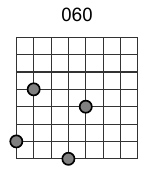}

\includegraphics[width=1.25cm]{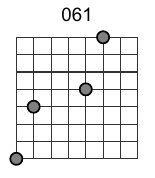}
\includegraphics[width=1.25cm]{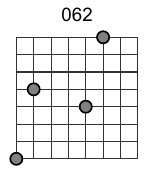}
\includegraphics[width=1.25cm]{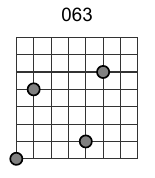}
\includegraphics[width=1.25cm]{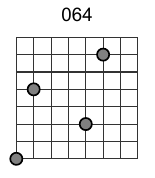}
\includegraphics[width=1.25cm]{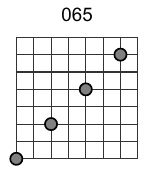}
\includegraphics[width=1.25cm]{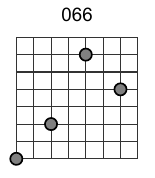}
\includegraphics[width=1.25cm]{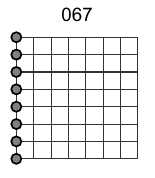}
\includegraphics[width=1.25cm]{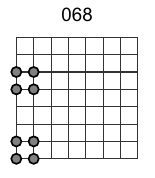}
\includegraphics[width=1.25cm]{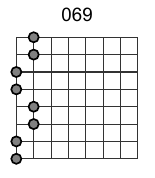}
\includegraphics[width=1.25cm]{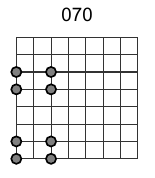}

\includegraphics[width=1.25cm]{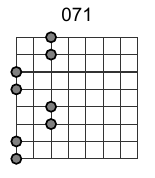}
\includegraphics[width=1.25cm]{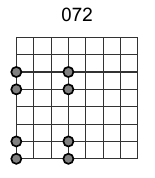}
\includegraphics[width=1.25cm]{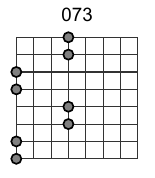}
\includegraphics[width=1.25cm]{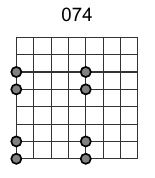}
\includegraphics[width=1.25cm]{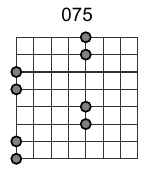}
\includegraphics[width=1.25cm]{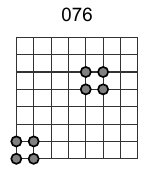}
\includegraphics[width=1.25cm]{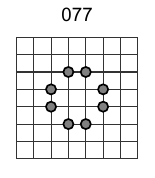}
\includegraphics[width=1.25cm]{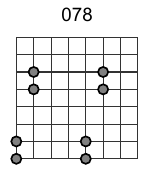}
\includegraphics[width=1.25cm]{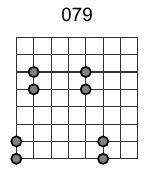}
\includegraphics[width=1.25cm]{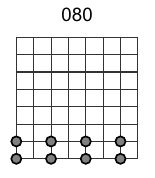}

\includegraphics[width=1.25cm]{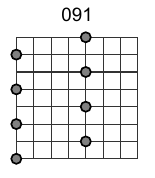}
\includegraphics[width=1.25cm]{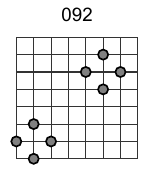}
\includegraphics[width=1.25cm]{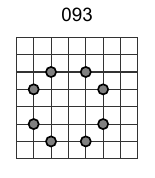}\includegraphics[width=1.25cm]{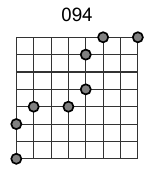}
\includegraphics[width=1.25cm]{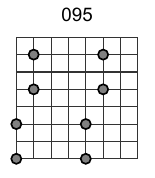}
\includegraphics[width=1.25cm]{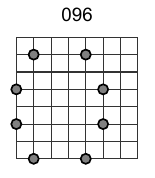}
\includegraphics[width=1.25cm]{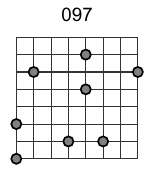}
\includegraphics[width=1.25cm]{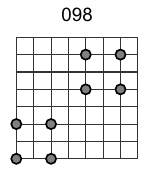}
\includegraphics[width=1.25cm]{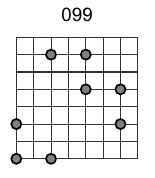}
\includegraphics[width=1.25cm]{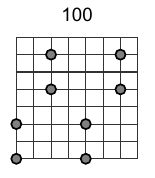}

\includegraphics[width=1.25cm]{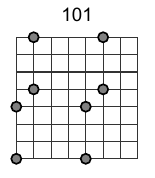}
\includegraphics[width=1.25cm]{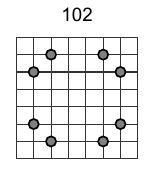}
\includegraphics[width=1.25cm]{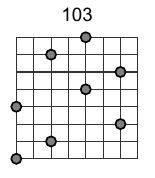}\includegraphics[width=1.25cm]{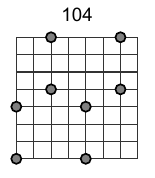}
\includegraphics[width=1.25cm]{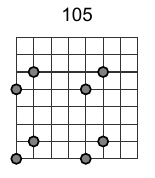}
\includegraphics[width=1.25cm]{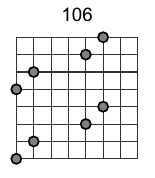}
\includegraphics[width=1.25cm]{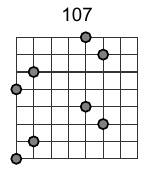}
\includegraphics[width=1.25cm]{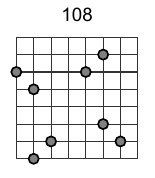}
\includegraphics[width=1.25cm]{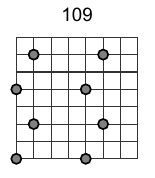}
\includegraphics[width=1.25cm]{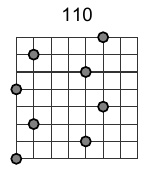}

\includegraphics[width=1.25cm]{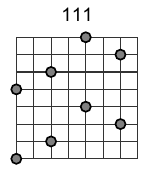}
\includegraphics[width=1.25cm]{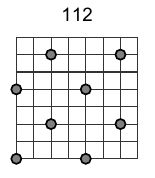}
\includegraphics[width=1.25cm]{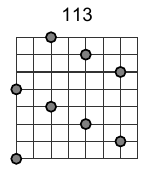}\includegraphics[width=1.25cm]{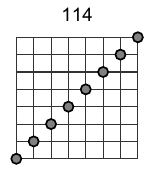}
\includegraphics[width=1.25cm]{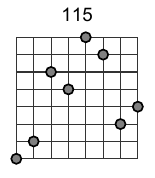}
\includegraphics[width=1.25cm]{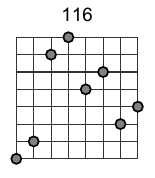}
\includegraphics[width=1.25cm]{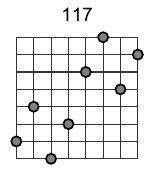}
\includegraphics[width=1.25cm]{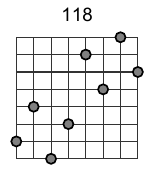}
\includegraphics[width=1.25cm]{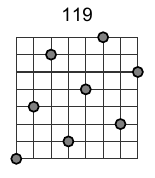}
\includegraphics[width=1.25cm]{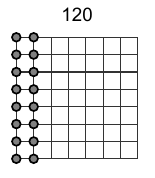}

\includegraphics[width=1.25cm]{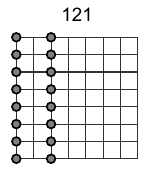}
\includegraphics[width=1.25cm]{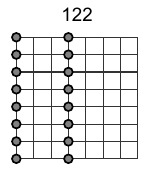}
\includegraphics[width=1.25cm]{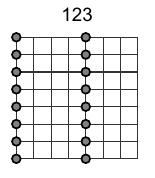}\includegraphics[width=1.25cm]{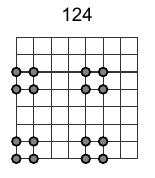}
\includegraphics[width=1.25cm]{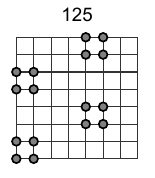}
\includegraphics[width=1.25cm]{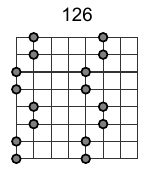}
\includegraphics[width=1.25cm]{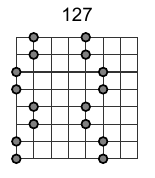}
\includegraphics[width=1.25cm]{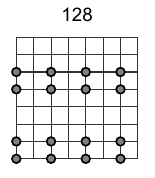}
\includegraphics[width=1.25cm]{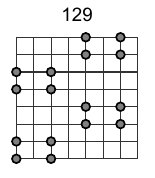}
\includegraphics[width=1.25cm]{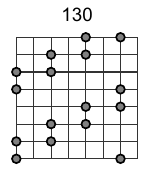}

\includegraphics[width=1.25cm]{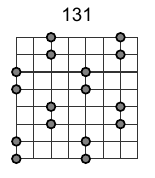}
\includegraphics[width=1.25cm]{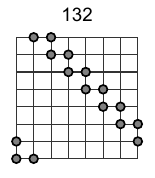}
\includegraphics[width=1.25cm]{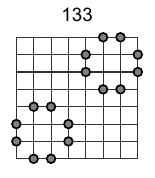}\includegraphics[width=1.25cm]{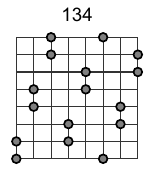}
\includegraphics[width=1.25cm]{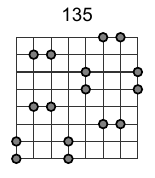}
\includegraphics[width=1.25cm]{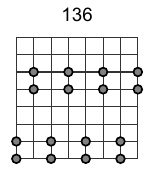}
\includegraphics[width=1.25cm]{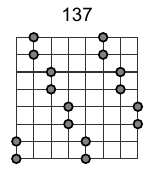}
\includegraphics[width=1.25cm]{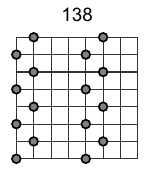}
\includegraphics[width=1.25cm]{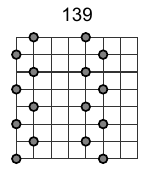}
\includegraphics[width=1.25cm]{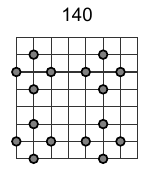}

\includegraphics[width=1.25cm]{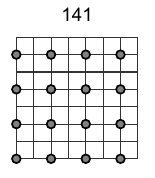}
\includegraphics[width=1.25cm]{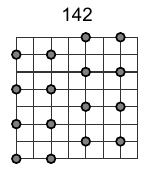}
\includegraphics[width=1.25cm]{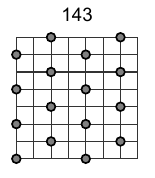}\includegraphics[width=1.25cm]{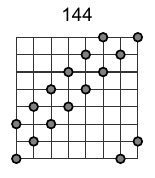}
\includegraphics[width=1.25cm]{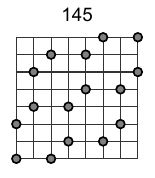}
\includegraphics[width=1.25cm]{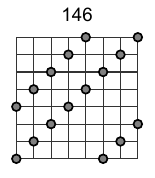}
\includegraphics[width=1.25cm]{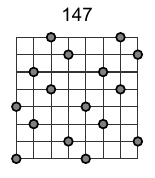}
\includegraphics[width=1.25cm]{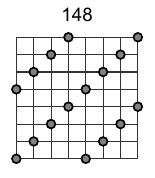}
\includegraphics[width=1.25cm]{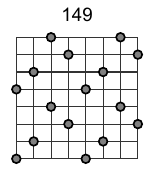}
\includegraphics[width=1.25cm]{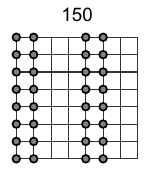}

\includegraphics[width=1.25cm]{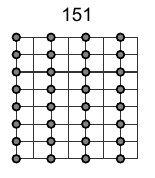}
\includegraphics[width=1.25cm]{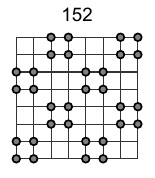}
\includegraphics[width=1.25cm]{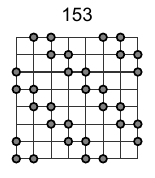}\includegraphics[width=1.25cm]{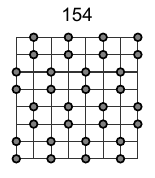}
\includegraphics[width=1.25cm]{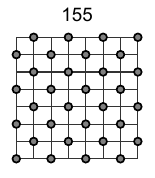}
\includegraphics[width=1.25cm]{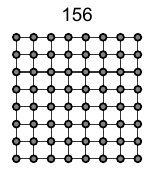}\!\phantom{\includegraphics[width=1.25cm]{fig/figs8/sq/147.pdf}
\includegraphics[width=1.25cm]{fig/figs8/sq/148.pdf}
\includegraphics[width=1.25cm]{fig/figs8/sq/149.pdf}
\includegraphics[width=1.25cm]{fig/figs8/sq/140.pdf}}
 \end{minipage}
\caption{The invariant equilibria for a square grid economy with $8\times 8$ locations.}
\label{fig:invs_sq_8}
\end{figure}

\begin{figure}
\centering

\includegraphics[width=2.2cm]{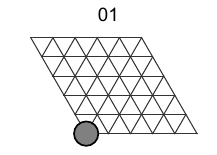}
\includegraphics[width=2.2cm]{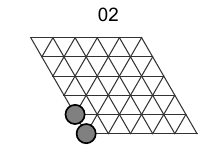}
\includegraphics[width=2.2cm]{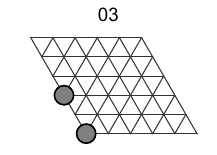}
\includegraphics[width=2.2cm]{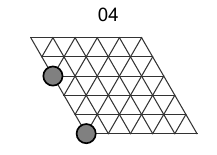}
\includegraphics[width=2.2cm]{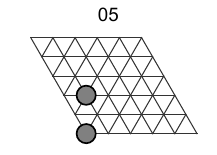}
\includegraphics[width=2.2cm]{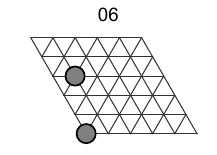}
\includegraphics[width=2.2cm]{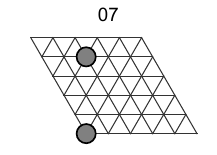}

\includegraphics[width=2.2cm]{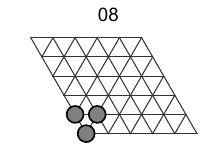}
\includegraphics[width=2.2cm]{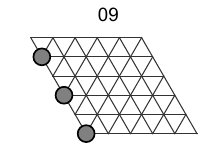}
\includegraphics[width=2.2cm]{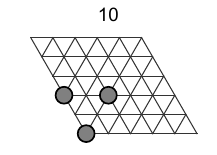}
\includegraphics[width=2.2cm]{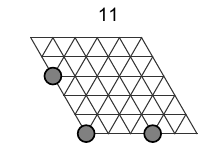}
\includegraphics[width=2.2cm]{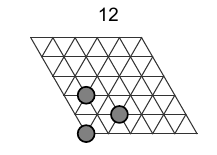}
\includegraphics[width=2.2cm]{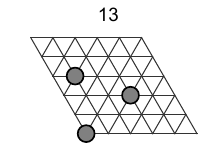}
\includegraphics[width=2.2cm]{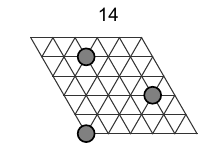}

\includegraphics[width=2.2cm]{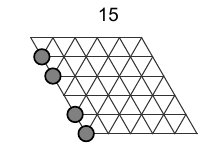}
\includegraphics[width=2.2cm]{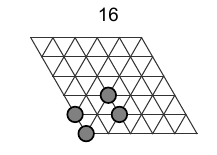}
\includegraphics[width=2.2cm]{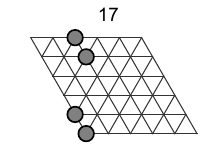}
\includegraphics[width=2.2cm]{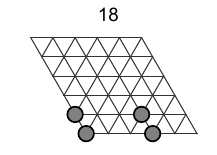}
\includegraphics[width=2.2cm]{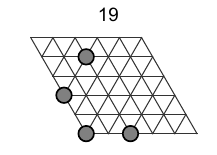}
\includegraphics[width=2.2cm]{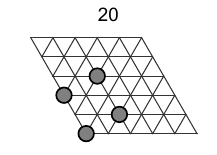}
\includegraphics[width=2.2cm]{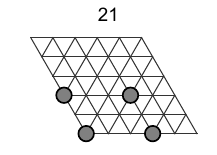}

\includegraphics[width=2.2cm]{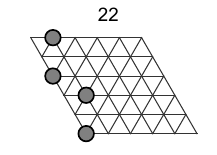}
\includegraphics[width=2.2cm]{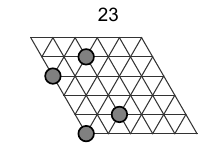}
\includegraphics[width=2.2cm]{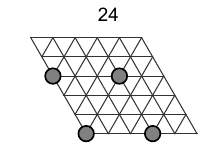}
\includegraphics[width=2.2cm]{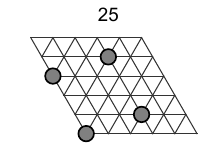}
\includegraphics[width=2.2cm]{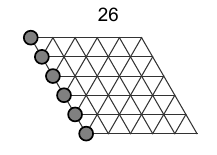}
\includegraphics[width=2.2cm]{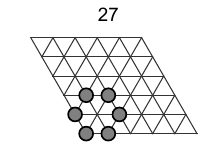}
\includegraphics[width=2.2cm]{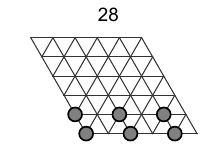}

\includegraphics[width=2.2cm]{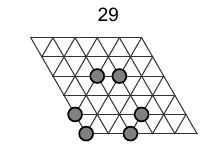}
\includegraphics[width=2.2cm]{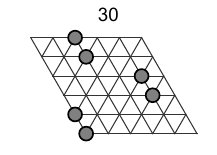}
\includegraphics[width=2.2cm]{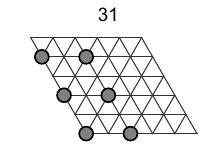}
\includegraphics[width=2.2cm]{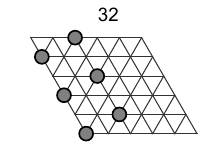}
\includegraphics[width=2.2cm]{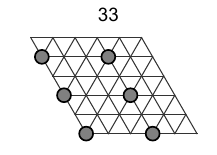}
\includegraphics[width=2.2cm]{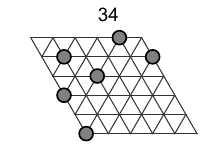}
\includegraphics[width=2.2cm]{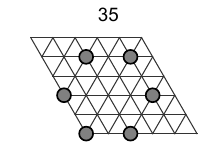}

\includegraphics[width=2.2cm]{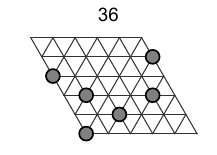}
\includegraphics[width=2.2cm]{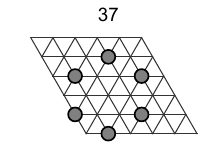}
\includegraphics[width=2.2cm]{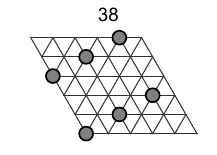}
\includegraphics[width=2.2cm]{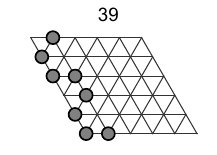}
\includegraphics[width=2.2cm]{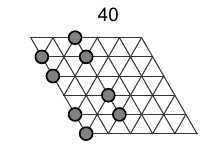}
\includegraphics[width=2.2cm]{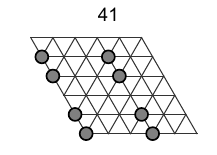}
\includegraphics[width=2.2cm]{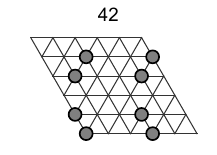}

\includegraphics[width=2.2cm]{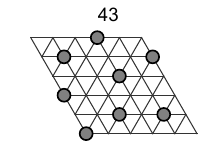}
\includegraphics[width=2.2cm]{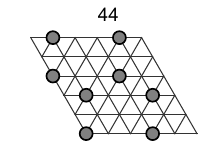}
\includegraphics[width=2.2cm]{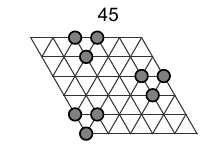}
\includegraphics[width=2.2cm]{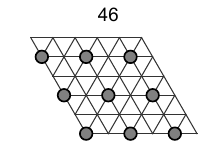}
\includegraphics[width=2.2cm]{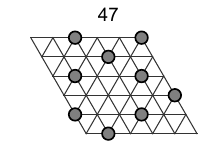}
\includegraphics[width=2.2cm]{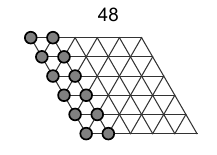}
\includegraphics[width=2.2cm]{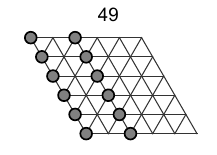}

\includegraphics[width=2.2cm]{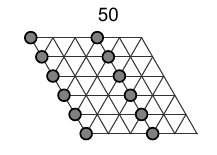}
\includegraphics[width=2.2cm]{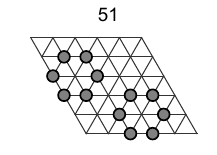}
\includegraphics[width=2.2cm]{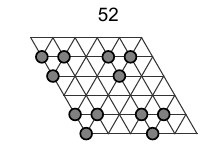}
\includegraphics[width=2.2cm]{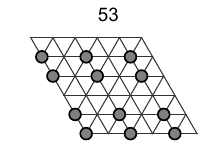}
\includegraphics[width=2.2cm]{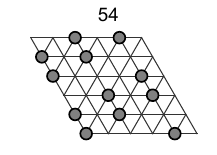}
\includegraphics[width=2.2cm]{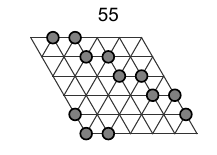}
\includegraphics[width=2.2cm]{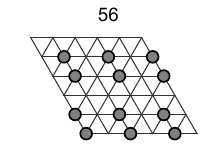}

\includegraphics[width=2.2cm]{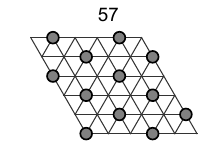}
\includegraphics[width=2.2cm]{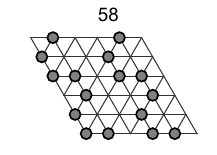}
\includegraphics[width=2.2cm]{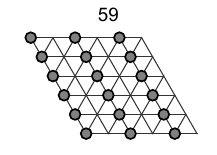}
\includegraphics[width=2.2cm]{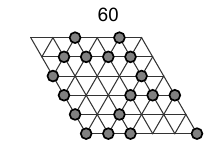}
\includegraphics[width=2.2cm]{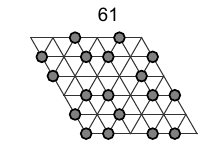}
\includegraphics[width=2.2cm]{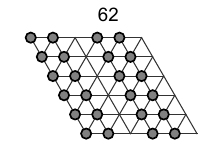}
\includegraphics[width=2.2cm]{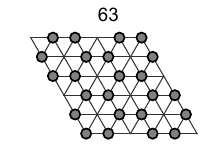}

\includegraphics[width=2.2cm]{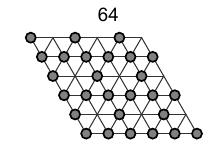}
\includegraphics[width=2.2cm]{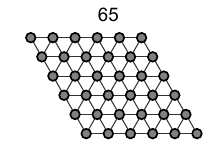}\!
\phantom{\includegraphics[width=2.2cm]{fig/hex/59.pdf}
\includegraphics[width=2.2cm]{fig/hex/60.pdf}
\includegraphics[width=2.2cm]{fig/hex/61.pdf}
\includegraphics[width=2.2cm]{fig/hex/62.pdf}
\includegraphics[width=2.2cm]{fig/hex/63.pdf}}
 \caption{The invariant equilibria for a triangular grid economy with $6\times 6$ locations.}
\label{fig:invs_hx}
\end{figure}

\clearpage 

{
\small
\singlespacing
\bibliographystyle{apalike}
\bibliography{refs}
}

\end{document}